\documentclass[11pt,reqno]{amsart}
\usepackage{amsmath}
\usepackage{txfonts}
\usepackage{mathrsfs}
\usepackage{amsfonts}
\allowdisplaybreaks[4]
\usepackage{graphicx} 
\usepackage{epsfig}
\usepackage{bbm}
\usepackage{amssymb}
\usepackage{appendix}
\usepackage{cite,color}
\usepackage{tikz}
\usepackage{esint}
\usepackage{bm}
\newtheorem{theorem}{Theorem}

\newtheorem{proposition}[theorem]{Proposition}
\newtheorem{corollary}[theorem]{Corollary}

\newtheorem{remark}[theorem]{Remark}
\newtheorem{lemma}[theorem]{Lemma}

\newcommand{\be}{\begin{equation}}
\newcommand{\ee}{\end{equation}}
\newcommand{\bea}{\begin{eqnarray}}
\newcommand{\eea}{\end{eqnarray}}
\newcommand{\ba}{\begin{array}}
	\newcommand{\ea}{\end{array}}
\newcommand{\bean}{\begin{eqnarray*}}
	\newcommand{\eean}{\end{eqnarray*}}

\newcommand{\La}{\Lambda}
\newcommand{\la}{\lambda}

\newcommand{\De}{\Delta}
\newcommand{\p}{\partial}

\newcommand{\E}{\mathcal{E}}

\begin{document}

\title{The two--component discrete KP hierarchy}
\author{Wenqi Cao$^1$, Jipeng Cheng $^{1,2*}$, Jinbiao Wang$^1$}
\dedicatory { $^{1}$School of Mathematics, China University of
Mining and Technology, Xuzhou, Jiangsu 221116,  China\\
$^{2}$ Jiangsu Center for Applied Mathematics (CUMT), \ Xuzhou, Jiangsu 221116, China}
\thanks{*Corresponding author. Email: chengjp@cumt.edu.cn, chengjipeng1983@163.com.}
\begin{abstract}
The discrete KP hierarchy is also known as the $(l-l')$--th modified KP hierarchy. Here in this paper, we consider the corresponding two--component generalization, called the two--component discrete KP (2dKP) hierarchy. Firstly, starting from the bilinear equation of the 2dKP hierarchy, we derive the corresponding Lax equation by the Shiota method, this is using scalar Lax operators involving two difference operators $\La_1$ and $\La_2$. Then starting from the 2dKP Lax equation, we obtain the corresponding bilinear equation, including the existence of the tau function. From above discussions, we can determine which are essential in the 2dKP Lax formulation. Finally, we discuss the reduction of the 2dKP hierarchy corresponding to the loop algebra $\widehat{sl}_{M+N}=sl_{M+N}[\la,\la^{-1}]\oplus\mathbb{C}c \ (M,N\geq1)$.  

\end{abstract}

\maketitle

\section{Introduction}
Since the 1980s, the KP theory \cite{Date1983,Dickey2003,Mulase1994,Harnad2021,Willox2004} has become an important research topic in mathematical physics and integrable systems, and made it possible to understand solitons and integrable
systems through a unified approach. Among the KP theory, here we are more interested in the discrete KP hierarchy \cite{Dickey1999LMP, Adler1999-1},
\begin{align}\label{zl*l' tau_l*tau_l'}
{\rm Res}_zz^{l-l'}\tau_l(x-[z^{-1}])\tau_{l'}(x'+[z^{-1}]e^{\xi(x-x',z)}=0,\quad\ l\geq l',
\end{align}
where
${\rm Res}_z\sum_{i} a_iz^i=a_{-1},\ [z^{-1}]=(z^{-1},z^{-2}/2, z^{-3}/3,\cdots),\ \xi(x,z)=\sum_{i\geq1}x_iz^i$ and $x=(x_1,x_2,\cdots)$.
The discrete KP hierarchy is also known as the $(l-l')$--th mKP hierarchy \cite{Jimbo1983,Kac2018JJP}. Note that the $0$--th mKP hierarchy is the usual KP hierarchy\cite{Date1983,Dickey2003,Willox2004,Mulase1994}, and the $1$--st mKP hierarchy is the Kupershmidt--Kiso mKP hierarchy \cite{Chengjipeng2018JGP,Konopelchenko1993,Kupershmidt1985, Kiso1990}.
The discrete KP hierarchy is of great importance. For example, the discrete KP hierarchy can be used to describe the Darboux orbits of the KP hierarchy \cite{Willox2004,Yangyi2022}. For the discrete KP hierarchy \eqref{zl*l' tau_l*tau_l'}, there are usually two different Lax formulations as follows.
\begin{itemize}
  \item The first one is expressed by the pseudo-differential operator $L_0=\partial+\sum_{i=0}^{\infty}a_i\partial^{-1}$ and a sequence of functions $\{v_i(x)\}_{i\in\mathbb{Z}}$ satisfying the following equations \cite{Dickey1999LMP,Dickey2003,Harnad2021}
  \begin{align}
\partial_{x_k}L_i=[(L_i^k)_{\geq0},L_i],\quad
\partial_{x_k}v_i=(L_{i+1}^k)_{\geq0}(\partial+v_{i})-(\partial+v_{i})(L_{i}^k)_{\geq0} \label{particalvi},
\end{align}
where $\partial=\partial_{x_1}$ and $L_{i}$ is defined by $L_{i+1}(\partial+v_i)=(\partial+v_i)L_i.$
\item The second one is characterized by \cite{Adler1999-1,Dickey1999LMP} 
\begin{align}\label{particalL=LkLL=La}
\partial_{x_k}L=[(L^k)_{\geq0},L],\quad L=\Lambda+\sum_{i=0}^{\infty}a_i\Lambda^{-i},
\end{align}
 where $\Lambda$\ is the shift operator acting on $f(n)$\ by
$\Lambda(f(n))=f(n+1)$ and $a_i=a_i(n,x)$. 

\end{itemize}

\noindent Here we believe the second Lax formulation \eqref{particalL=LkLL=La} of the discrete KP hierarchy is more convenient, since the shift operator $\La$ can more easily link the different discrete variables.

Next in this paper, we will discuss the two--component generalization of the discrete KP hierarchy. Usually, the multi--component extension of the KP theory can be obtained through the multi--component boson--fermion correspondence\cite{Jimbo1983,Kac2023,Kac2003,Takebe2024}. Firstly, the fermionic form of the discrete KP hierarchy (\ref{zl*l' tau_l*tau_l'}) is given by \cite{Jimbo1983,Kac2018JJP}
\begin{align}\label{s taul taul'}
S(\tau_l\otimes\tau_{l'})=0,\quad\tau_l\in\mathcal F_l,\quad\ l\geq l',
\end{align}
where $S=\sum_{i\in\mathbb{Z}+1/2}\psi_i^+\otimes\psi_i^-$ \ with the charged free fermions $\psi_i^{\pm}$ $(i\in\mathbb{Z}+1/2)$ satisfying \cite{Jimbo1983,Kac2023,Kac2003}
$$\psi_i^{\lambda}\psi_j^{\mu}+\psi_j^{\mu}\psi_i^{\lambda}=\delta_{\lambda+\mu,0}\delta_{i+j,0},\quad  \lambda,\mu=\pm.$$ $\mathcal F_l$ is the subspace of the Fock space $\mathcal F$ with the charge $l$. And the Fock space $\mathcal F$ is the vector space spanned by $\psi_{i_1}^+\psi_{i_2}^+\cdots\psi_{i_r}^+\psi_{j_1}^-\psi_{j_2}^-\cdots\psi_{j_s}^-|0\rangle,$ for $i_1<\cdots<i_{r}<0$ and $j_1<\cdots<j_{s}<0$, where the vacuum $|0\rangle$ is defined by $\psi_{i}^{\pm}|0\rangle=0\ (i>0)$. If set the $\text{charge of}\ \psi_k^{\pm}=\pm1$, then the charge of $\psi_{i_1}^+\psi_{i_2}^+\cdots\psi_{i_r}^+\psi_{j_1}^-\psi_{j_2}^-\cdots\psi_{j_s}^-|0\rangle $ is $r-s$. In particular, we have $\mathcal F=\bigoplus_{l\in\mathbb Z}\mathcal F_l.$ Then the two--component boson--fermion correspondence \cite{Jimbo1983,Kac2023,Kac2003} is given by $$\sigma: \mathcal{F}\cong
\mathbb{C}[Q_1,Q_2,t=(t^{(1)},t^{(2)})|t^{(a)}=(t_1^{(a)},t_2^{(a)},\cdots)],$$ where $Q_a$ commutes with $t_j^{(b)}$ and $Q_1Q_2=-Q_2Q_1$. And $\sigma$ is uniquely defined by $\sigma(|0\rangle)=1$ and
\begin{align*}
\sigma\psi^{\pm(a)}(z)\sigma^{-1}=Q_a^{\pm1}z^{\pm Q_a\partial_{Q_a}}
e^{\pm\xi(t^{(a)},z)}e^{\mp\xi(\widetilde{\partial}_{t^{(a)}},z^{-1})},\quad a=1,2,
\end{align*}
where $\widetilde{\partial}_{t^{(a)}}=(\partial_{t_1^{(a)}},\partial_{t_2^{(a)}}/2,\partial_{t_3^{(a)}}/3,\cdots)$, $\psi^{\pm(a)}(z)=\sum_{j \in \mathbb{Z}+1/2}\psi_{j}^{\pm(a)}z^{-j-1/2}$ and $\psi_{j}^{\pm(a)}$ can be defined by \cite{Kac2023,Kac2003}
\begin{align*}
&\psi_{Mi+p+1/2}^{+(1)}=\psi^+_{(M+N)i+p+1/2}, \ \quad \psi_{Mi-p-1/2}^{-(1)}=\psi^-_{(M+N)i-p-1/2}, \quad 0\leq p\leq M-1,\\
&\psi_{Ni+q+1/2}^{+(2)}=\psi^+_{(M+N)i+M+q+1/2},\quad \psi_{Ni-q-1/2}^{-(2)}=\psi^-_{(M+N)i-M-q-1/2}, \quad 0\leq q\leq N-1.
\end{align*}
For $\tau_l\in\mathcal F_l$ satisfying \eqref{s taul taul'}, if define
\begin{align*}
\sigma(\tau_l)=\sum_{m\in\mathbb{Z}}(-1)^{\frac{m(m-1)}{2}}
Q_1^{m}Q_2^{l-m}\tau_{m,m-l}(t^{(1)},t^{(2)}), 
\end{align*}
then the fermionic discrete KP hierarchy \eqref{s taul taul'} can be converted to the following form
\begin{align}\label{bartau z(m-m')}
&\oint_{C_R}\frac{dz}{2\pi i}z^{m_1-m_1'}e^{\xi(t^{(1)}-t^{(1)'},z)}\tau_{m_1,m_2}(t-[z^{-1}]_1)\tau_{m_1',m_2'}(t+[z^{-1}]_1)\nonumber\\
=&\oint_{C_r}\frac{dz}{2\pi i}z^{m_2-m_2'}e^{\xi(t^{(2)}-t^{(2)'},z^{-1})}\tau_{m_1+1,m_2+1}(t-[z]_2)\tau_{m_1'-1,m_2'-1}(t+[z]_2),\quad m_1-m_2\geq m_1'-m_2',
\end{align}
where $[z^{-1}]_1=([z^{-1}],0)$, $[z]_2=(0,[z])$, $C_R$ means the anticlockwise circle $|z|=R$ for sufficient large R, and $C_r$ is the anticlockwise circle $|z|=r$ with $r$  sufficient small. This equation \eqref{bartau z(m-m')} is called the two--component discrete KP hierarchy (2dKP for short), also known as the two--component mKP hierarchy \cite{Jimbo1983,Teo2011,Takebe2024}, which is a special case of the generalized two--component mKP hierarchy (see (4.4) in \cite{Jimbo1983}) and the $3$--KP hierarchy in \cite{Cui2025}.

As far as we can know, there is no unified approach to derive Lax equations from bilinear equations, which is still an open problem \cite{Kac1989,Kac1996}. For the 2dKP hierarchy, the Lax equation can be derived by matrix pseudo--differential operators (see (4.9) in \cite{Teo2011}), just like the first Lax formulation \eqref{particalvi} for the discrete KP hierarchy. Here, we would like to construct the Lax equation of the 2dKP hierarchy, similar to the second formulation \eqref{particalL=LkLL=La} of the discrete KP hierarchy by using Shiota method \cite{Shiota1989}, where we use the scalar Lax operators involving two difference operators $\Lambda_1$ and $\Lambda_2$ satisfying $\Lambda_a(f(\bm{m}))=f(\bm{m}+\bm{e}_a)$ with $\bm{e}_1=(1,0)$ and $\bm{e}_2=(0,1)$. Specifically, if introduce Lax operators
\begin{align*}
L_1(\bm{m},\Lambda_1)=\La_1+\sum_{i=0}^{+\infty}u_i^{(1)}(\bm{m})\Lambda_1^{-i},\quad L_2(\bm{m},\Lambda_2)=u_{-1}^{(2)}(\bm{m})\La_2^{-1}+\sum_{i=0}^{+\infty}u_i^{(2)}(\bm{m})\La_2^i,
\end{align*}
and a special operator $H=\Delta_2\Lambda_1+\rho$ with $\Delta_2=\Lambda_2-1$, then the Lax equation of the 2dKP hierarchy \eqref{bartau z(m-m')} is given by
\begin{align*}
&\partial_{t_k^{(a)}}L_a(\bm{m},\Lambda_a)=[B_k^{(a)}(\bm{m},\Lambda_a),L_a(\bm{m},\Lambda_a)],\quad \partial_{t_k^{(a)}}L_{3-a}(\bm{m},\Lambda_{3-a})=[\pi_{3-a}(B_k^{(a)}(\bm{m},\Lambda_{a})),L_{3-a}(\bm{m},\Lambda_{3-a})], \\
&H(\bm{m})L_1(\bm{m},\Lambda_a)=L_1(\bm{m}+\bm{e},\Lambda_1)H(\bm{m}),\quad H(\bm{m})L_2(\bm{m},\Lambda_2)=\Delta_2^{*}L_2(\bm{m}+\bm{e})\Delta_2^{*-1}H(\bm{m}),\\
&\p_{t_k^{(a)}}H(\bm{m})=C_k^{(a)}(\bm{m},\Lambda_a)H(\bm{m})-H(\bm{m})\cdot B_k^{(a)}(\bm{m},\Lambda_a),\quad a=1,2,
\end{align*}
where $\bm{e}=(1,1)$, $\Delta_a^*=\Lambda_a^{-1}-1$, $\Delta_a^{*-1}=\sum_{j=1}^{+\infty}\Lambda_a^{j}$, and 
\begin{align*}
&B_k^{(1)}(\bm{m},\Lambda_1)=\Big(L_1^k(\bm{m},\Lambda_1)\Big)_{1,\geq 0}, \quad B_k^{(2)}(\bm{m},\Lambda_2)=\Big(L_2^k(\bm{m},\Lambda_2)\Big)_{\Delta_2^*,\geq 1},\\
&C_k^{(1)}(\bm{m},\Lambda_1)=B_k^{(1)}(\bm{m}+\bm{e},\La_1),\quad C_k^{(2)}(\bm{m},\Lambda_2)=\Delta_2^{*}\cdot B_k^{(2)}(\bm{m}+\bm{e},\La_2)\cdot\Delta_2^{*-1},\\
&\pi_1(\La_2^{-k})=\prod_{j=1}^k\Big(1-\iota_{\Lambda_1^{-1}}
\big(\La_1-\rho(\bm{m}-j\bm{e}_2)\big)^{-1}\cdot\rho(\bm{m}-j\bm{e}_2)\Big),\\
&\pi_2(\La_1^k)=(-1)^k\prod_{j=1}^k\Big(\iota_{\Lambda_2}(\La_2-1)^{-1}\cdot\rho(\bm{m}+(j-1)\bm{e}_1)\Big). 
 \end{align*}  
Here $(\sum_ia_i\Lambda_1^i)_{1,\geq0}=\sum_{i\geq0}a_i\Lambda_1^i$, $(\sum_ia_i\Delta_2^{*i})_{\Delta_2^{*},\geq1}=\sum_{i\geq1}a_i\Delta_2^{*i}$, and $\iota_{\Lambda_a^{\pm1}}f(\Lambda_a)$ means expanding $f(\Lambda_a)$ in the form of $\sum_{\mp j\ll+\infty}a_j\Lambda_a^j$.  

The derivation of the Lax equation for the 2dKP hierarchy from the bilinear equation \eqref{bartau z(m-m')} is quite similar to the case of the 3--KP hierarchy in \cite{Cui2025}, except the part involving H. But there is no discussion from the Lax equation to the bilinear equation in \cite{Cui2025}, which is usually more difficult. Therefore, besides the investigation from the bilinear equation to the Lax equation, we will focus on how to obtain the bilinear equation \eqref{bartau z(m-m')} from the Lax equation for the 2dKP hierarchy, which cotains the following key steps:
\begin{itemize}
\item from Lax equation to wave operators;
\item from wave operators to bilinear equation;
\item from bilinear equation to the existence of the tau function.
\end{itemize}
In fact, the result from the Lax equation to the bilinear equationcan help us to fix which are more basic in the 2dKP Lax formulation. Finally, we discuss the reduction of the 2dKP hierarchy corresponding to the loop algebra $\widehat{sl}_{M+N}=sl_{M+N}[\lambda,\lambda^{-1}]\oplus\mathbb{C}c$, called the $\widehat{sl}_{M+N}$--reduction of the 2dKP hierarchy, which is given by the operator $H=\La_1\Delta_2+\rho$ and the Lax operator $$\mathcal{L}=\La_1^M+\sum_{k=0}^{M-1}u_k\La_1^k+\sum_{l=1}^{N}v_l(\La_2^{-l}-1)$$ satisfying
\begin{align*}
 &H\mathcal{L}=(C_M^{(1)}+C_N^{(2)})H,\quad \partial_{t_k^{(a)}}H=C_k^{(a)}H-HB_k^{(a)},\\
  &\partial_{t_k^{(a)}}\mathcal{L}=[\pi_1(B_k^{(a)}),\pi_1(\mathcal{L})]_{1,\geq0}+[\pi_2(B_k^{(a)}),\pi_2(\mathcal{L})]_{\Delta_2^*,\geq1},\\ \text{or} \ &\partial_{t_k^{(a)}}\pi_b(\mathcal{L})=[\pi_b(B_k^{(a)}),\pi_b(\mathcal{L})],\quad a,b=1,2.
\end{align*}
Here $B_k^{(1)}=\Big(\pi_1(\mathcal{L})^{\frac{k}{M}}\Big)_{1,\geq 0},$ $B_k^{(2)}=\Big(\pi_2(\mathcal{L})^{\frac{k}{N}}\Big)_{\Delta_2^*,\geq1}$ and
   $C_k^{(1)}(\bm{m})=B_k^{(1)}(\bm{m}+\bm{e}),$ $C_k^{(2)}(\bm{m})=\Delta_2^*B_k^{(2)}(\bm{m}+\bm{e})\Delta_2^{*-1}.$

The remaining of this paper is organized in the way below. In Section 2, we derive the Lax equation of the 2dKP hierarchy from the corresponding bilinear equation. Then in Section 3, we consider the inverse direction, that is from the Lax equation to the bilinear equation, including the existence of the tau function. Next in Section 4, we discuss the reduction of the 2dKP hierarchy corresponding to the loop algebra $\widehat{sl}_{M+N}$. Finally, some conclusions and discussions are given in Section 5.

\section{From bilinear equation to Lax equation}
In this section, we will start from the bilinear equation \eqref{bartau z(m-m')} of the 2dKP hierarchy, and obtain the corresponding Lax equations by introducing wave functions, wave operators and Lax operators. The method here is quite similar to the case of the 3--KP hierarchy in \cite{Cui2025}, except the contents involving the operator $H=\La_1\Delta_2+\rho$ and the projections $\pi_a$. 

If introduce the wave functions $\Psi_{a}(\bm{m},t,z)$ and the adjoint wave functions $\widetilde{\Psi}_{a}(\bm{m},t,z)$ as follows
\begin{align}\label{psi_1}
&\Psi_{1}(\bm{m},t,z)=z^{m_1}e^{\xi(t^{(1)},z)}\frac{\tau_{\bm{m}}(t-[z^{-1}]_1)}{\tau_{\bm{m}}(t)},\nonumber\\
&\widetilde{\Psi}_{1}(\bm{m},t,z)=z^{-m_1+1}e^{-\xi(t^{(1)},z)}\frac{\tau_{\bm{m}}(t+[z^{-1}]_1)}{\tau_{\bm{m}}(t)},\nonumber\\
&\Psi_{2}(\bm{m},t,z)=z^{m_2}e^{\xi(t^{(2)},z^{-1})}\frac{\tau_{\bm{m}+\bm{e}}(t-[z]_2)}{\tau_{\bm{m}}(t)},\nonumber\\
&\widetilde{\Psi}_{2}(\bm{m},t,z)=z^{-m_2+1}e^{-\xi(t^{(2)},z^{-1})}\frac{\tau_{\bm{m}-\bm{e}}(t+[z]_2)}{\tau_{\bm{m}}(t)},
\end{align}
then the bilinear equation \eqref{bartau z(m-m')} of the 2dKP hierarchy can be written into
\begin{align}\label{oint_{C_R}}
\oint_{C_R}\frac{dz}{2\pi iz}\Psi_{1}(\bm{m},t,z)\widetilde{\Psi}_{1}(\bm{m}',t',z)=\oint_{C_r}\frac{dz}{2\pi iz}\Psi_{2}(\bm{m},t,z)\widetilde{\Psi}_{2}(\bm{m}',t',z), \quad m_1-m_2\geq m_1'-m_2'.
\end{align}

Before further discussion, let us introduce some symbols. For the formal operator $A=\sum_{j_1,j_2\in\mathbb{Z}}a_{j_1,j_2}(\bm{m})\Lambda_1^{j_1}\Lambda_2^{j_2},$  let us denote
$$A^*=\sum_{j_1,j_2\in\mathbb{Z}}\Lambda_1^{-j_1}\Lambda_2^{-j_2}a_{j_1,j_2}(\bm{m}),\  A_{a,P}=\sum_{\text{$j_a$ satisfies $P$},j_{3-a}\in\mathbb{Z}}a_{j_1,j_2}(\bm{m})\Lambda_1^{j_1}\Lambda_2^{j_2},\ A_{a,[k]}=\sum_{j_a=k,j_{3-a}\in\mathbb{Z}}a_{j_1,j_2}(\bm{m})\Lambda_1^{j_1}\Lambda_2^{j_2},$$
where $a=1,2$, $P\in\{\geq k, \leq k,>k,<k\}$ with $k\in\mathbb{Z}$. Further if set $\Delta_a=\Lambda_a-1$ and $\Delta_a^*=\Lambda_a^{-1}-1,$ then $\Delta_a^{-1}$ means $\sum_{j=1}^{+\infty}\Lambda_a^{-j}$, while $\Delta_a^{*-1}$ means $\sum_{j=1}^{+\infty}\Lambda_a^{j}$.
Also for $R\in\{\Delta_1,\Delta_2,\Delta_1^*,\Delta_2^*\}$, we set $(R+1)^{-k}=\sum_{j=0}^\infty \binom{-k}{j}R^{-k-j}$ and $(\sum b_jR^j)_{R,\geq k}=\sum_{j\geq k}b_jR^j$, where $\binom{-k}{j}=\frac{(-k)(-k-1)\cdots(-k-j+1)}{j!}$. 

\begin{lemma}\cite{Adler1999-1}\label{ABlemma}
Let $A(m,\Lambda)=\sum_{j}a_j(m)\Lambda^j,\ B(m,\Lambda)=\sum_{j}b_j(m)\Lambda^j$ be two operators with shift operator $\Lambda$ defined by $\Lambda(f(m))=f(m+1)$, where $a_j(\bm{m})=b_j(\bm{m})=0$ for $j\gg0$ $(or j \ll0)$, then
\begin{align*}
A(m,\Lambda)\cdot B(m,\Lambda)^*=\sum_{j\in\mathbb{Z}}{\rm Res}_zz^{-1}\left(A(m,\Lambda)(z^{\pm m})\cdot
B(m+j,\Lambda) (z^{\mp m\mp j})\right)\Lambda^j,
\end{align*}
\end{lemma}
After the preparation above, if set $\bm{m}'=\bm{m}+\bm{j}$ with $\bm{j}=(j_1,j_2)$, then \eqref{oint_{C_R}} can be written as
\begin{align}\label{sum_{j}}
&\sum_{j_1,j_2\in \mathbb{Z}}\left(\oint_{C_R}\frac{dz}{2\pi iz}\Psi_{1}(\bm{m},t,z)\widetilde{\Psi}_{1}(\bm{m}+\bm{j},t',z)\Lambda_{1}^{j_1}\right)_{1,\leq j_2}\Lambda_{2}^{j_2}\nonumber\\
=&\sum_{j_1,j_2\in \mathbb{Z}}\left(\oint_{C_r}\frac{dz}{2\pi iz}\Psi_{2}(\bm{m},t,z)\widetilde{\Psi}_{2}(\bm{m}+\bm{j},t',z)\Lambda_{2}^{j_2}\right)_{2,\geq j_1}\Lambda_{1}^{j_1}.
\end{align}
Next let us introduce wave operators $S_a$ and $\widetilde{S}_a(a=1,2)$ as follows,
\begin{equation}\label{SaandwideSa}
\begin{aligned}
&S_1(\bm{m},t,\Lambda_1)=1+\sum_{k=1}^{+\infty}a_k^{(1)}\Lambda_{1}^{-k},\quad\quad\quad\ \widetilde{S}_1(\bm{m},t,\Lambda_1)=1+\sum_{k=1}^{+\infty}\widetilde{a}_k^{(1)}\Lambda_{1}^k,\\ &S_2(\bm{m},t,\Lambda_2)=a_0^{(2)}+\sum_{k=1}^{+\infty}a_k^{(2)}\Lambda_{2}^k,\quad\quad \widetilde{S}_2(\bm{m},t,\Lambda_2)=\widetilde{a}_0^{(2)}+\sum_{k=1}^{+\infty}\widetilde{a}_k^{(2)}\Lambda_{2}^{-k}.
\end{aligned}
\end{equation}
satisfying
\begin{equation}\label{S_{1}e}
\begin{aligned}
&\Psi_{1}(\bm{m},t,z)=S_1(\bm{m},t,\Lambda_1)(z^{m_1})e^{\xi(t^{(1)},z)},\quad\ 
\widetilde{\Psi}_{1}(\bm{m},t,z)=\widetilde{S}_1(\bm{m},t,\Lambda_1)(z^{-m_1})e^{-\xi(t^{(1)},z)}z,\\
&\Psi_{2}(\bm{m},t,z)=S_2(\bm{m},t,\Lambda_2)(z^{m_2})e^{\xi(t^{(2)},z^{-1})},\quad \widetilde{\Psi}_{2}(\bm{m},t,z)=\widetilde{S}_2(\bm{m},t,\Lambda_2)(z^{-m_2})e^{-\xi(t^{(2)},z^{-1})}z.
\end{aligned}
\end{equation}
In particular by \eqref{psi_1}, we can find $a_0^{(2)}(\bm{m})=\frac{\tau_{\bm{m}+\bm{e}}}{\tau_{\bm{m}}}$ and $\widetilde{a}_0^{(2)}(\bm{m})=\frac{\tau_{\bm{m}-\bm{e}}}{\tau_{\bm{m}}}$.

\begin{proposition}\label{conclu S12}
The relations between $S_a$ and $\widetilde{S}_a$ $(a=1,2)$ are given by
 \begin{align*}
S_{1}(\bm{m},\Lambda_1)\Lambda_1{\widetilde{S}^*_{1}(\bm{m},\Lambda_1)}=\Lambda_1,\quad
S_2(\bm{m},\Lambda_2)\Lambda_2\widetilde{S}_2^*(\bm{m}+\bm{e}_1,\Lambda_2)=\Delta_2^{*-1},
\end{align*}
where  $A(\bm{m},\Lambda_a)$ means $A(\bm{m},t,\Lambda_a)$ for short.
\end{proposition}
\begin{proof}
If set $t'=t$, then by \eqref{S_{1}e}, we have 
\begin{align*}
&\sum_{j_1,j_2\in \mathbb{Z}}\left(\oint_{C_R}\frac{dz}{2\pi iz}S_{1}(\bm{m},\Lambda_1)(z^{m_1})\widetilde{S}_{1}(\bm{m}+\bm{j},\Lambda_1)(z^{-m_1-j_1})\Lambda_{1}^{j_1}\right)_{1,\leq j_2}\Lambda_{2}^{j_2}\nonumber\\
=&\sum_{j_1,j_2\in \mathbb{Z}}\left(\oint_{C_r}\frac{dz}{2\pi iz}S_{2}(\bm{m},\Lambda_2)(z^{m_2})\widetilde{S}_{2}(\bm{m}+\bm{j},\Lambda_2)(z^{-m_2-j_2})\Lambda_{2}^{j_2}\right)_{2,\geq j_1}\Lambda_{1}^{j_1}.
\end{align*}
Next by Lemma \ref{ABlemma}, we can obtain:
\begin{align*}
\sum_{j_2\in \mathbb{Z}}\left(S_{1}(\bm{m},\Lambda_1)\Lambda_{1}{\widetilde{S}^*_{1}(\bm{m}+j_2\bm{e}_2,\Lambda_1)}\right)_{1,\leq j_2}\Lambda_{2}^{j_2}=\sum_{j_1\in \mathbb{Z}}\left(S_{2}(\bm{m},\Lambda_2)\Lambda_{2}{\widetilde{S}^*_{2}(\bm{m}+j_1\bm{e}_1,\La_2)}\right)_{2,\geq j_1}\Lambda_{1}^{j_1}.
\end{align*} 
Then by comparing the coefficients of $\Lambda_2^0$ and $\Lambda_1$ respectively, the results can be obtained.
\end{proof}
\begin{proposition}\label{conclu particals1s2}
Evolution equations of wave operators $S_a$ with respect to $t_k^{(b)}\ (b=1,2)$ are given as follows
\begin{align*}
  &\partial_{t_k^{(1)}}S_1(\bm{m},\Lambda_1)
=B_k^{(1)}(\bm{m},\Lambda_1)S_1(\bm{m},\Lambda_1)-S_1(\bm{m},\Lambda_1)\Lambda_1^k,\\
&\partial_{t_k^{(1)}}S_2(\bm{m},\Lambda_2)
=\left(B_k^{(1)}(\bm{m},\Lambda_1)S_1(\bm{m},\Lambda_1)\Delta_2^{*-1}S_1^{-1}(\bm{m},\Lambda_1)\right)_{1,[0]}\Delta_2^*S_2(\bm{m},\Lambda_2),\\
 &\partial_{t_k^{(2)}}S_1(\bm{m},\Lambda_1)
=\left(B_k^{(2)}(\bm{m},\Lambda_2)S_2(\bm{m},\Lambda_2)\Delta_1^{-1}
S_2^{-1}(\bm{m},\Lambda_2)\Delta_2^{*-1}\right)_{2,[0]}S_1(\bm{m},\Lambda_1),\\
&\partial_{t_k^{(2)}}S_2(\bm{m},\Lambda_2)=B_k^{(2)}(\bm{m},\Lambda_2)S_2(\bm{m},\Lambda_2)-S_2(\bm{m},\Lambda_2)\Lambda_2^{-k},
  \end{align*}
where $B_k^{(1)}\big(\bm{m},\Lambda_1)=(S_1(\bm{m},\Lambda_1)\Lambda_1^kS_1(\bm{m},\Lambda_1)^{-1}\big)_{1,\geq0}$ and $B_k^{(2)}(\bm{m},\Lambda_2)=\big(S_2(\bm{m},\Lambda_2)\Lambda_2^{-k}S_2(\bm{m},\Lambda_2)\big)_{\Delta_2^*,\geq1}$.
\end{proposition}
\begin{proof}
If apply $\partial_{t_k^{(1)}}$ to both sides of \eqref{sum_{j}}, and let  $t'=t$, then we can get by  Lemma \ref{ABlemma} that
\begin{align*}
&\sum_{j_2\in\mathbb{Z}}\left(\Big(\partial_{t_k^{(1)}}S_1(\bm{m},\Lambda_1)+S_1(\bm{m},\Lambda_1)\Lambda_1^k\Big)\cdot S_1^{-1}(\bm{m}+j_2\bm{e}_2,\Lambda_1)\Lambda_1\right)_{1,\leq j_2}\Lambda_2^{j_2}\nonumber\\
=&\sum_{j_1\in\mathbb{Z}}\left(\partial_{t_k^{(1)}}S_2(\bm{m},\Lambda_2)\cdot
S_2^{-1}(\bm{m}+(j_1-1)\bm{e}_1,\Lambda_2)\Delta_2^{*-1}\right)_{2,\geq j_1}\Lambda_1^{j_1}.
\end{align*}
Then by comparing the coefficients of $\Lambda_2^0$ and $\La_1$, we can get the results for $\partial_{t_k^{(1)}}S_a$. Similarly, one can get $\partial_{t_k^{(2)}}S_a$.
\end{proof}
\begin{proposition}\label{conclulamba1s1lamba2s2}
 Given $k>0$, the actions of $\La_a^{(3-2a)k}$ on $S_{3-a}$ are given by
  \begin{align*}
&\Lambda_1^k(S_2(\bm{m},\Lambda_2))=\left(\Lambda_1^kS_1(\bm{m},\Lambda_1)\Delta_2^{*-1}
S_1^{-1}(\bm{m},\Lambda_1)\right)_{1,[0]}\cdot\Delta_2^*S_2(\bm{m},\Lambda_2),\\
&\Lambda_2^{-k}(S_1(\bm{m},\Lambda_1))=\left(\left(\Lambda_2^{-k}S_2(\bm{m},\Lambda_2)\Delta_1^{-1}
S_2^{-1}(\bm{m},\Lambda_2)\Delta_2^{*-1}\right)_{2,[0]}+1\right)\cdot S_1(\bm{m},\Lambda_1).
\end{align*}
\end{proposition}
\begin{proof}
Firstly by setting $\bm{m}\rightarrow \bm{m}+k\bm{e}_1$ in \eqref{oint_{C_R}} and using Lemma \ref{ABlemma}, Proposition \ref{conclu S12}, we have
\begin{align*}
&\sum_{j_2\in\mathbb{Z}}\left(S_1(\bm{m}+k\bm{e}_1,\Lambda_1)\Lambda_1^{k}{S}_1^{-1}(\bm{m}+j_2\bm{e}_2,\Lambda_1)\Lambda_1\right)_{1,\leq j_2}\Lambda_2^{j_2}\\
=&\sum_{j_1\in\mathbb{Z}}\left(S_2(\bm{m}+k\bm{e}_1,\Lambda_2){S}_2^{-1}(\bm{m}+(j_1-1)\bm{e}_1,\Lambda_2)\Delta_2^{*-1}\right)_{2,\geq j_1}\Lambda_1^{j_1}.
\end{align*} 
By comparing the coefficients of $\Lambda_1$, we can get the result for $\Lambda_1^k(S_2)$. Similarly, by setting $\bm{m}\rightarrow \bm{m}-k\bm{e}_2$ 
in \eqref{oint_{C_R}}, we can get $\Lambda_2^{-k}(S_1)$.
\end{proof}
Note that if set $k=1$ in Proposition \ref{conclulamba1s1lamba2s2}, then
\begin{align}
&\Lambda_1(S_2(\bm{m},\Lambda_2))=-\Delta_2^{*-1}\Lambda_2\cdot\rho(\bm{m})\cdot S_2(\bm{m},\Lambda_2), \label{lambda1S2}\\
&\Lambda_2^{-1}(S_1(\bm{m},\Lambda_1))=(1-\Lambda_1^{-1}\cdot\rho(\bm{m}))\cdot S_1(\bm{m},\Lambda_1), \label{lambda2S1}
\end{align}
where $\rho(\bm{m},t)=\frac{\tau_{\bm{m}}}{\tau_{\bm{m}+\bm{e}_1}}\frac{\tau_{\bm{m}+\bm{e}+\bm{e}_1}}{\tau_{\bm{m}+\bm{e}}}=\partial_{t_1^{(1)}}\log\frac{\tau_{\bm{m}+\bm{e}}}{\tau_{\bm{m}+\bm{e_1}}}$. Here we have used the fact that $\tau_{\bm{m}}$ satisfies $D_{t_1^{(1)}}\tau_{\bm{m}+\bm{e}_2}\cdot\tau_{\bm{m}}=\tau_{\bm{m}+\bm{e}}\tau_{\bm{m}-\bm{e}_1}$, with $D_{t_1^{(1)}}$ being Hirota derivative.
Thus if denote $$H=\La_1\De_2+\rho,$$then we have the corollary below.
\begin{corollary}\label{concluH}
The operator H is related with the wave operators $S_1$ and $S_2$ by 
\begin{align}\label{H=Lam1=lam2}
H=-\Lambda_1\Lambda_2S_1\Delta_2^{*}S_1^{-1}=-\Lambda_1\Delta_2S_2\Delta_1^{*}S_2^{-1}.
\end{align} 
Evolutions equations of $H$ are given by
\begin{align*}
\p_{t_k^{(a)}}H(\bm{m})=C_k^{(a)}(\bm{m})H(\bm{m})-H(\bm{m})\cdot B_k^{(a)}(\bm{m}),\quad a=1,2,
\end{align*}
where $C_k^{(1)}(\bm{m})=B_k^{(1)}(\bm{m}+\bm{e})$ and $C_k^{(2)}(\bm{m})=\Delta_2^{*}\cdot B_k^{(2)}(\bm{m}+\bm{e})\cdot\Delta_2^{*-1}.$
\end{corollary}
\begin{proof}
Firstly by \eqref{lambda1S2} and \eqref{lambda2S1}, we have
\begin{align}\label{HS1S2}
        H(\bm{m})\cdot S_1(\bm{m},\La_1)=(\La_1-\rho(\bm{m}))\cdot S_1(\bm{m},\La_1)\cdot \De_2,\quad H(\bm{m})\cdot S_2(\bm{m},\La_2)=-\rho(\bm{m}) S_2(\bm{m},\La_2)\cdot \De_1.
    \end{align}
Then by substituting $\rho=H-\Lambda_1\Delta_2$ into $H(\bm{m})\cdot S_1(\bm{m},\La_1)=(\La_1-\rho(\bm{m}))\cdot S_1(\bm{m},\La_1)\cdot \De_2$, we have $HS_2=-(H-\Lambda_1\Delta_2)S_2\Delta_1$, i.e., $H=\Lambda_1\Delta_2S_2\Delta_1\Lambda_1^{-1}S_2^{-1}$. Another one can be obtained similarly. Finally, $\p_{t_k^{(a)}}H(\bm{m})$ can be obtained by \eqref{H=Lam1=lam2} and $\partial_{t_k^{(1)}}S_1(\bm{m},\Lambda_1)$, $\partial_{t_k^{(2)}}S_2(\bm{m},\Lambda_2)$ in Proposition \ref{conclu particals1s2}.
\end{proof}

\begin{corollary}\label{partial psi}
The wave functions $\Psi_a$ and the adjoint wave functions $\widetilde{\Psi}_a$ satisfy the following relations,
\begin{align*}
&\partial_{t_k^{(1)}}\Psi_a(\bm{m})=B_k^{(1)}(\bm{m},\Lambda_1)\Big(\Psi_a(\bm{m})\Big),
\quad \partial_{t_k^{(2)}}\Psi_a(\bm{m})=B_k^{(2)}(\bm{m},\Lambda_2)\Big(\Psi_a(\bm{m})\Big),\\
&\partial_{t_k^{(1)}}\widetilde{\Psi}_a(\bm{m})=-B_k^{*(1)}(\bm{m}-\bm{e}_1,\Lambda_1)\left(\widetilde{\Psi}_a(\bm{m})\right),\quad \partial_{t_k^{(2)}}\widetilde{\Psi}_a(\bm{m})=-\Delta_2^{-1}B_k^{*(2)}(\bm{m}-\bm{e}_1,\Lambda_2)\Delta_2\left(\widetilde{\Psi}_a(\bm{m})\right),\\
&H(\bm{m})(\Psi_a(\bm{m})) = 0,\quad \widetilde{H}(\bm{m})(\widetilde{\Psi}_a(\bm{m}+\bm{e}_1)) = 0,
\end{align*}
where $\Psi_a(\bm{m})=\Psi_a(m,t,z)$, $\widetilde{\Psi}_a(\bm{m})=\widetilde{\Psi}_a(\bm{m},t,z)$ and $\widetilde{H}(\bm{m})=H^*(\bm{m}-\bm{e})$.
\end{corollary}
\begin{proof}
For $\partial_{t_k^{(a)}}\Psi_a(\bm{m})\ (a=1,2)$, they can be obtained directly by $\partial_{t_k^{(a)}}S_a$ in Proposition \ref{conclu particals1s2} and the definitions of $\Psi_a$ in \eqref{S_{1}e}. From $\partial_{t_k^{(1)}}S_2(\bm{m},\Lambda_2)$ in Proposition \ref{conclu particals1s2} and $\Lambda_1^{l}(S_2)$ in Proposition \ref{conclulamba1s1lamba2s2}, we can deduce that 
\begin{align}\label{particaltk1 S2}
\partial_{t_k^{(1)}}S_2(\bm{m},\Lambda_2)=B_k^{(1)}(\bm{m},\Lambda_1)\Big(S_2(\bm{m},\Lambda_2)\Big),
\end{align}
 and thus $\partial_{t_k^{(1)}}\Psi_2(\bm{m})=B_k^{(1)}(\bm{m},\Lambda_1)\Big(\Psi_2(\bm{m})\Big)$. Similarly, it follows that $\partial_{t_k^{(2)}}\Psi_1(\bm{m})=B_k^{(2)}(\bm{m},\Lambda_2)\Big(\Psi_1(\bm{m})\Big).$

By the similar methods in Proposition \ref{conclu particals1s2} and Proposition \ref{conclulamba1s1lamba2s2}, we can get 
\begin{align*}
&\partial_{t_k^{(1)}}\widetilde{S}^*_{2}(\bm{m},\Lambda_2)=-\Lambda_1^{-1}\Lambda_2^{-1}S_2^{-1}(\bm{m},\Lambda_2)\left(S_1(\bm{m},\Lambda_1)\Delta_2^{*-1}S_{1}^{-1}(\bm{m},\Lambda_1)B^{(1)}_{k}(\bm{m},\Lambda_1)\Lambda_1\right)_{1,[1]},\\
&\partial_{t_k^{(2)}}\widetilde{S}^*_{1}(\bm{m},\Lambda_1)=-\Lambda_1^{-1}S_1^{-1}(\bm{m},\Lambda_1)\left(S_2(\bm{m},\Lambda_2)\Delta_1^{-1} S_{2}^{-1}(\bm{m},\Lambda_2)B^{(2)}_{k}(\bm{m},\Lambda_2)\Delta_2^{*-1}\Lambda_1\right)_{2,[0]}, \end{align*}
and
\begin{align}
&\Lambda_1^{-k}(\widetilde{S}^*_2(\bm{m},\Lambda_2))=\Lambda_1^{-1}\Lambda_2^{-1}S_2^{-1}(\bm{m},\Lambda_2)\left(S_1(\bm{m},\Lambda_1)\Delta_2^{*-1}S_1^{-1}(\bm{m},\Lambda_1)\Lambda_1^{k+1})\right)_{1,[1]},\label{lam1s1star}\\
&\Lambda_2^k(\widetilde{S}^*_1(\bm{m},\Lambda_1))=\Lambda_1^{-1}S_1^{-1}(\bm{m},\Lambda_1)\left(\Lambda_1+(S_2(\bm{m},\Lambda_2)\Delta_1^{-1}S^{-1}_2(\bm{m},\Lambda_2)\Lambda_2^{-k}\Delta_2^{*-1}\Lambda_1)_{2,[0]}\right).\label{lam2s2star}
\end{align}
Based upon this, the results about $\partial_{t_k^{(a)}}\widetilde{\Psi}_i(\bm{m})$ can be similarly obtained just like $\p_{t_k^{(a)}}\Psi_i$. $H(\bm{m})(\Psi_i(\bm{m})) = 0$ and $\widetilde{H}(\bm{m})(\Psi_i(\bm{m}+\bm{e}_1)) = 0$ can be obtained by \eqref{S_{1}e} and \eqref{H=Lam1=lam2}.
\end{proof}

In order to express the Lax equation of the 2dKP hierarchy, we need to further introduce some new symbols. Let us define for $a=1,2$, 
\begin{alignat*}{2}
    \E_{(a)}=\mathcal{B}[\La_{3-a},\La_{3-a}^{-1}]((\La_a^{2a-3})),\quad \E_{(a)}^{0}=\mathcal{B}((\La_a^{2a-3})),
\end{alignat*}
where $\mathcal{B}$ is the set of the functions depending on $\bm{m}$ and $t$. Then we have the proposition below.
\begin{proposition}\label{prop: sumdecom}
$\E_{(a)}=\E_{(a)}^{0}\oplus\E_{(a)}H,\quad \text{for} \quad a =1,2.$
\end{proposition}
\begin{proof}
Here we only prove $\E_{(1)}=\E_{(1)}^{0}\oplus\E_{(1)}H$, since another is almost the same. Firstly for $A\in\E_{(1)}^{0}\bigcap\E_{(1)}H$, let us assume $$A=\sum_{i\leq M}\sum_{j=-N_1}^ {N_2}b_{i,j}\Lambda_1^i\Lambda_2^jH,\quad N_1, N_2\geq0,$$ then we can find that 
\begin{align*}
A=\sum_{i\leq M}b_{i,N_2}\Lambda_1^{i+1}\Lambda_2^{N_2+1}+\sum_{j=-N_1}^{N_2}\sum_{i\leq M}(b_{i,j-1}-b_{i,j}+b_{i+1,j}\cdot\rho_{ij})\Lambda_1^{i}\Lambda_2^j,
\end{align*}
where we assume $b_{i,-N_1-1}=b_{M+1,j}=0$ and $\rho_{ij}=\rho(m_1+i+1,m_2+j)$. Further $A\in\E_{(1)}^0$ implies that the coefficients of $\Lambda_2^j(j\neq0)$ vanish, that is
\begin{align}
  &b_{i,N_2}=0,\label{bin2=0}\\
  &b_{i,j-1}-b_{i,j}+b_{i+1,j}\cdot\rho_{ij}=0,\quad i\leq M,\  -N_1\leq j\leq N_2(j\neq0),\label{bij-1bij}\\
  &b_{i,-N_1}-b_{i+1,-N_1}\cdot\rho_{i,-N_1}=0,\label{b1N1}
\end{align}
where $i\leq M$, $-N_1\leq j\leq N_2$, $j\neq0$.
Thus from \eqref{bin2=0} and \eqref{bij-1bij}, we can know $b_{ij}=0$ for $0\leq j\leq N_2$, $i\leq M$. Further by \eqref{b1N1}, we know $b_{M,-N_1}=0$. So the successive applications of \eqref{b1N1}, implies $b_{i,-N_1}=0$, for $i\leq M.$ Next set $j=-N_1+1$ in \eqref{bij-1bij}, we know 
\begin{align}
b_{i,-N_1+1}=b_{i+1,-N_1+1}\cdot\rho_{i,-N_1+1}.\label{b_i,N_1+1}
\end{align}
Notice that $b_{M,-N_1+1}=0$ by setting $i=M$ in \eqref{b_i,N_1+1}. So by \eqref{b_i,N_1+1}, $b_{i,-N_1+1}=0$, for $i\leq M.$
Continue above discussions, we can get $b_{i,j}=0$ for $i\leq M$, $-N_1\leq j\leq-1$. Therefore $A=0$, which means  $\E_{(1)}^{0}\bigcap\E_{(1)}H=\{0\}$.

Finally we just need to prove that
$\E_{(1)}\subseteq\E_{(1)}^{0}+\E_{(1)}H$,
which means
\begin{align}
 \{\La_1^i\La_2^j\mid i\leq M,\ -N_1\leq j\leq N_2\}\subseteq \E_{(1)}^{0}+\E_{(1)}H,\quad M\in\mathbb{Z}, \ N_1,N_2\in\mathbb{Z}_{\geq0}. \label{esubseteqeoeh} 
\end{align}
 Since $\La_1^i\in \E_{(1)}^0$ for $i\leq M$,
we next make induction on $j$ to complete the proof. Assuming \eqref{esubseteqeoeh} holds for $j>0$, we will prove it for $j+1$, i.e. $\La_2\cdot\La_1^i\La_2^{j}\in\E_{(1)}^0\oplus\E_{(1)}H$.
By hypothesis $\La_1^i\La_2^{j}=\sum_{l\leq N}a_l\La_1^l+PH$ for $a_l\in \mathcal{B}$ and $P\in \E_{(1)}$,
we have
$$\La_2\cdot\La_1^i\La_2^{j}=\sum_{l\leq N}a_l(\bm{m}+\bm{e}_2)\cdot\La_1^{l-1}\cdot\big(H-\rho+\La_1\big)+\La_2\cdot PH \in\E_{(1)}^0\oplus\E_{(1)}H,
$$
where we have used $\La_2=\La_1^{-1}\cdot(H-\rho(\bm{m}))+1$.
While the case for $j<0$ is similar. So we finish the proof.
\end{proof}
\begin{remark}
The proof of Proposition \ref{prop: sumdecom} only depends on the definition of $\E_{(a)}$, $\E_{(a)}^0$ and $H$, while the proof in \cite{Cui2025} relies on $H(\Psi_a)=0.$
\end{remark}
Due to Proposition \ref{prop: sumdecom}, we can define the following projections 
$$\pi_a:\E_{(a)}=\E_{(a)}^0\oplus\E_{(a)}H\to \E_{(a)}^{0}, \quad a =1,2,$$
and we can find the following recursion relations for $\pi_a$ by the above definition,
\begin{align}\label{pilamk+1}
     \pi_{3-a}(\La_a^{(-1)^a(k+1)})=\La_a(\pi_{3-a}(\La_a^{(-1)^ak}))\cdot\pi_{3-a}(\La_a^{(-1)^a}).
  \end{align}
Next from $H=\Lambda_1\Delta_2+\rho$, we can know
\begin{align*}
 &\pi_1(\La_2^{-1})=1-\iota_{\Lambda_1^{-1}}
 \big(\La_1-\rho(\bm{m}-\bm{e}_2)\big)^{-1}\cdot\rho(\bm{m}-\bm{e}_2),\\
 &\pi_2(\La_1)=-\iota_{\Lambda_2}(\La_2-1)^{-1}\cdot\rho(\bm{m}).
  \end{align*}
So we can get for $k>0$,
\begin{align}
&\pi_1(\La_2^{-k})=\prod_{j=1}^k\Big(1-\iota_{\Lambda_1^{-1}}
\big(\La_1-\rho(\bm{m}-j\bm{e}_2)\big)^{-1}\cdot\rho(\bm{m}-j\bm{e}_2)\Big),\label{pi1La1}\\
&\pi_2(\La_1^k)=(-1)^k\prod_{j=1}^k\Big(\iota_{\Lambda_2}(\La_2-1)^{-1}\cdot\rho(\bm{m}+(j-1)\bm{e}_1)\Big). \label{pi2la12}
 \end{align} 
 \begin{lemma}\label{Lemma:pi1pi2la}
The projections $\pi_a$ $(a=1,2)$ can also be computed by using $H$ in the way below
\begin{align}\label{lemmapi1pi2}
\pi_2(\La_1^k)=(\Lambda_1^k\cdot\iota_{\La_1^{-1}}H^{-1}\cdot\Lambda_1)_{1,[0]}\Delta_2,\quad \pi_1(\La_2^{-k})=-(\Lambda_2^{-k}\cdot\iota_{\La_2}H^{-1}\cdot\Lambda_2)_{2,[0]}\Lambda_1. 
 \end{align} 
\end{lemma}
\begin{proof}
We give a proof for the first equation, and the second equation can be proved in a similar way. Firstly, if we set
\begin{align}\label{H-1La}
  \iota_{\La_1^{-1}}H^{-1}(\bm{m})\cdot\Lambda_1=\iota_{\Lambda_2}(\La_2-1)^{-1}+\sum_{j=1}^{\infty}\La_1^{-j}\cdot v_j(\bm{m}),
\end{align}
then comparing the coefficients of $\La^{i}$ for
$$\La_1=H(\bm{m})\cdot\Big(\iota_{\Lambda_2}(\La_2-1)^{-1}+\sum_{j=1}^\infty\La_1^{-j}\cdot v_j(\bm{m})\Big)=(\La_1\Delta_2+\rho(\bm{m}))\Big(\iota_{\Lambda_2}(\La_2-1)^{-1}+\sum_{j=1}^{\infty}\La_1^{-j}\cdot v_j(\bm{m})\Big),$$
we have $$v_1(\bm{m})=-\iota_{\Lambda_2}(\La_2-1)^{-1}\cdot\rho(\bm{m})\cdot\iota_{\Lambda_2}(\La_2-1)^{-1},\quad v_{l+1}(\bm{m})=-\iota_{\Lambda_2}(\La_2-1)^{-1}\cdot\rho(\bm{m}+l\bm{e}_1)\cdot v_l(\bm{m}).$$
Thus, we can get
\begin{align}\label{vk2}
  v_k(\bm{m})&=-\iota_{\Lambda_2}(\La_2-1)^{-1}\cdot\rho(\bm{m}+(k-1)\bm{e}_1)\cdot v_{k-1}(\bm{m})\nonumber \\
  &=\iota_{\Lambda_2}(\La_2-1)^{-1}\cdot\rho(\bm{m}+(k-1)\bm{e}_1)\cdot\iota_{\Lambda_2}(\La_2-1)^{-1}\cdot\rho(\bm{m}-(k-2)\bm{e}-1)\cdot v_{k-2}(\bm{m}) \nonumber\\
  &=\cdots \nonumber\\
  &=(-1)^k\prod_{i=1}^{k}\iota_{\Lambda_2}(\La_2-1)^{-1}\cdot\rho(\bm{m}+(k-i)\bm{e}_1)\cdot\iota_{\Lambda_2}(\La_2-1)^{-1}.
\end{align}
Next, taking \eqref{H-1La} into the first equation of \eqref{lemmapi1pi2}, we can obtain 
\begin{align}\label{pi2La1k=vkde2}
\pi_2(\La_1^k)=\Big(\La_1^k\big(\iota_{\Lambda_2}(\La_2-1)^{-1}+\sum_{j\geq1}\La_1^{-j}v_j(\bm{m})\big)\Big)_{1,[0]}\Delta_2=v_k(\bm{m})\cdot\Delta_2.
\end{align}
Finally, combining \eqref{vk2} and \eqref{pi2La1k=vkde2}, we get 
\begin{align*}
\pi_2(\La_1^k)=(-1)^k\prod_{j=1}^k\Big(\iota_{\Lambda_2}(\La_2-1)^{-1}\cdot\rho(\bm{m}+(j-1)\bm{e}_1)\Big),
\end{align*}
which is just equation \eqref{pi2la12}.
\end{proof}
\begin{remark}
Note that in the proof of Lemma \ref{Lemma:pi1pi2la}, we only use the definitions of $H$ and $\pi_a$.
\end{remark}
\begin{corollary}\label{corollary pi1pi2}
In terms of wave operators $S_1$ and $S_2$,
\begin{align*}
&\pi_2(\La_1^k)=\left(\Lambda_1^kS_1(\bm{m},\Lambda_1)\Delta_2^{*-1}
S_1^{-1}(\bm{m},\Lambda_1)\right)_{1,[0]}\cdot\Delta_2^*,\\
&\pi_1(\La_2^{-k})=\left(\Lambda_2^{-k}S_2(\bm{m},\Lambda_2)\Delta_1^{-1}
S_2^{-1}(\bm{m},\Lambda_2)\Delta_2^{*-1}\right)_{2,[0]}+1. 
 \end{align*} 
 \begin{proof}
 In fact by \eqref{H=Lam1=lam2}, we can know
 \begin{align*}
   &S_1\Delta_2^{*-1}S_1^{-1}=-\iota_{\La_1^{-1}}H^{-1}\cdot\Lambda_1\Lambda_2 \\
   &S_2\Delta_1^{-1}S_2^{-1}=-(1+\iota_{\La_2}H^{-1}\cdot\Lambda_1\Lambda_2\Delta_2^*),
 \end{align*}
 which implies
 \begin{equation}\label{H-1debianxing}
 \begin{aligned}
   &\iota_{\La_1^{-1}}H^{-1}=S_1\Delta_2^{*-1}S_1^{-1}\Lambda_1^{-1}\Lambda_2^{-1} \\
   &\iota_{\La_2}H^{-1}=-(1+S_2\Delta_1^{-1}S_2^{-1})\Delta_2^{*-1}\Lambda_1^{-1}\Lambda_2^{-1}.
 \end{aligned}
 \end{equation}
 Insert \eqref{H-1debianxing} into \eqref{lemmapi1pi2}, we can proof this corollary.
 \end{proof}
\end{corollary}
After the preparation above, now let us define the corresponding Lax operator:
\begin{equation}\label{deofL1andL2}
\begin{aligned}
&L_1(\bm{m},\Lambda_1)=S_1(\bm{m},\Lambda_1)\cdot \Lambda_1\cdot S_1^{-1}(\bm{m},\Lambda_1)=\La_1+\sum_{i=0}^{+\infty}u_i^{(1)}(\bm{m})\Lambda_1^{-i},\\
&L_2(\bm{m},\Lambda_2)=S_2(\bm{m},\Lambda_2)\cdot \Lambda_2^{-1}\cdot S_2^{-1}(\bm{m},\Lambda_2)=u_{-1}^{(2)}(\bm{m})\La_2^{-1}+\sum_{i=0}^{+\infty}u_i^{(2)}(\bm{m})\La_2^i,
\end{aligned}
\end{equation}
then we have the following theorem.
\begin{theorem}\label{tkflowofLi}
$L_1$ and $L_2$ defined by \eqref{deofL1andL2} satisfy the following Lax equations,
\begin{align*}
  &\partial_{t_k^{(a)}}L_b(\bm{m},\Lambda_b)=[\pi_b(B_k^{(a)}(\bm{m},\Lambda_a)),L_b(\bm{m},\Lambda_b)],\\  &H(\bm{m})L_1(\bm{m},\Lambda_1)=L_1(\bm{m}+\bm{e},\Lambda_1)H(\bm{m}),\quad H(\bm{m})L_2(\bm{m},\Lambda_2)=\Delta_2^{*}L_2(\bm{m}+\bm{e})\Delta_2^{*-1}H(\bm{m}),\\ &\p_{t_k^{(a)}}H(\bm{m})=C_k^{(a)}(\bm{m},\Lambda_a)H(\bm{m})-H(\bm{m})\cdot B_k^{(a)}(\bm{m},\Lambda_a),\ a,b=1,2,
\end{align*} 
where $B_k^{(1)}(\bm{m},\Lambda_1)=\Big(L_1^k(\bm{m},\Lambda_1)\Big)_{1,\geq 0},$ $B_k^{(2)}(\bm{m},\Lambda_2)=\Big(L_2^k(\bm{m},\Lambda_2)\Big)_{\Delta_2^*,\geq 1},$ $C_k^{(1)}(\bm{m},\Lambda_1)=B_k^{(1)}(\bm{m}+\bm{e})$, $C_k^{(2)}(\bm{m},\Lambda_2)=\Delta_2^{*}\cdot B_2^k(\bm{m}+\bm{e})\cdot\Delta_2^{*-1}$.
\end{theorem}
\begin{proof}
Firstly, $\partial_{t_k^{(a)}}L_a(\bm{m},\Lambda_a)$ an be proved by Proposition \ref{conclu particals1s2}. While by Lemma Proposition \ref{conclu particals1s2} and Corollary \ref{corollary pi1pi2}, we know $\partial_{t_k^{(3-a)}}S_a\cdot S_a^{-1}=\pi_a(B_k^{(3-a)})$, which leads to $\partial_{t_k^{(3-a)}}L_a(\bm{m},\Lambda_a)$. As for the result for $HL_a$, it can be obtained by \eqref{H=Lam1=lam2}. Finally by \eqref{H=Lam1=lam2} and Proposition \ref{conclu particals1s2}, we can get $\partial_{t_k^{(a)}}H(\bm{m})$.
\end{proof}

\noindent\textbf{Example. }
Firstly
$B_1^{(1)}(\bm{m},\La_1)= \Lambda_1+u_0^{(1)}(\bm{m}),$ $B_1^{(2)}(\bm{m},\La_2)=u_{-1}^{(2)}(\bm{m})\La_2^{-1}+u_{-1}^{(2)}(\bm{m}),$
 then we have
 \begin{align*}
 &\pi_1\Big(B_1^{(2)}(\bm{m},\La_2)\Big)=-u_{-1}^{(2)}(\bm{m})\cdot\iota_{\Lambda_1^{-1}}
 \big(\La_1-\rho(\bm{m}-\bm{e}_2)\big)^{-1}\cdot\rho(\bm{m}-\bm{e}_2)\\
  &\pi_2\Big(B_1^{(1)}(\bm{m},\La_1)\Big)=\rho(\bm{m})+u_0^{(1)}(\bm{m})+\De_2^{*-1}\rho(\bm{m}).
\end{align*}
Thus we can get
\begin{align*}
  &\p_{t_1^{(1)}}\rho(\bm{m})=\rho(\bm{m})(u_0^{(1)}(\bm{m}+\bm{e})-u_0^{(1)}(\bm{m})),\\
    &\p_{t_1^{(2)}}\rho(\bm{m})=\rho(\bm{m})( u_{-1}^{(2)}(\bm{m}+\bm{e})-u_{-1}^{(2)}(\bm{m})),\\
  &\partial_{t_{1}^{(1)}}u_0^{(1)}(\bm{m})=u_1^{(1)}(\bm{m}+\bm{e}_1)-u_1^{(1)}(\bm{m}),\\
  &\partial_{t_{1}^{(1)}}u_{-1}^{(2)}(\bm{m})=u_0^{(1)}(\bm{m})u_{-1}^{(2)}(\bm{m})-u_{-1}^{(2)}(\bm{m})u_0^{(1)}(\bm{m}-\bm{e}_2),\\
  &\partial_{t_1^{(2)}}u_0^{(1)}(\bm{m})=u_{-1}^{(2)}(\bm{m})u_1^{(2)}(\bm{m}-\bm{e}_2)-u_1^{(2)}(\bm{m})u_{-1}^{(2)}(\bm{m}+\bm{e}_2)\\
  &\partial_{t_1^{(2)}}u_{-1}^{(2)}(\bm{m})=u_{-1}^{(2)}(\bm{m})u_0^{(2)}(\bm{m}-\bm{e}_2)+(u_{-1}^{(2)}(\bm{m}))^2-u_{-1}^{(2)}(\bm{m})u_{-1}^{(2)}(\bm{m}-\bm{e}_2)-u_{0}^{(2)}(\bm{m})u_{-1}^{(2)}(\bm{m}).
\end{align*}

\section{From the Lax equations To bilinear equations}
In this section, we will start from the Lax equation of the 2dKP hierarchy and derive the corresponding bilinear equation for the 2dKP hierarchy, which contains the following steps:
\begin{align*}
 \text{Lax equation}\Longrightarrow \text{wave operator}\Longrightarrow \text{bilinear equation}\Longrightarrow \text{existence of tau function}.
\end{align*}

Firstly, the Lax equations of the 2dKP hierarchy are given by the following Lax operators
\begin{align*}
L_1(\bm{m},\Lambda_1)=\La_1+\sum_{i=0}^{+\infty}u_i^{(1)}(\bm{m})\Lambda_1^{-i},\quad L_2(\bm{m},\Lambda_2)=u_{-1}^{(2)}(\bm{m})\La_2^{-1}+\sum_{i=0}^{+\infty}u_i^{(2)}(\bm{m})\La_2^i,
\end{align*}
and one special operator $H(\bm{m})=\Lambda_1\Delta_2+\rho(\bm{m})$ satisfying
\begin{align}
  &\partial_{t_k^{(a)}}L_b(\bm{m},\Lambda_b)=[\pi_b(B_k^{(a)}(\bm{m},\Lambda_a)),L_b(\bm{m},\Lambda_b)],\label{particalLb} \\ &H(\bm{m})L_1(\bm{m},\Lambda_1)=L_1(\bm{m}+\bm{e},\Lambda_1)H(\bm{m}),\quad H(\bm{m})L_2(\bm{m},\Lambda_2)=\Delta_2^{*}L_2(\bm{m}+\bm{e})\Delta_2^{*-1}H(\bm{m}), \label{Hl1L2} \\ &\p_{t_k^{(a)}}H(\bm{m})=C_k^{(a)}(\bm{m},\Lambda_a)H(\bm{m})-H(\bm{m})\cdot B_k^{(a)}(\bm{m},\Lambda_a),\ a,b=1,2,\label{particalH}
\end{align}
where the projection $\pi_a$: $\E_{(a)}=\E_{(a)}^0\oplus\E_{(a)}H\to\E_{(a)}^0$ can be computed by Lemma \ref{Lemma:pi1pi2la}, and
\begin{align*}
&B_k^{(1)}(\bm{m},\Lambda_1)=\Big(L_1^k(\bm{m},\Lambda_1)\Big)_{1,\geq 0},\quad B_k^{(2)}(\bm{m},\Lambda_2)=\Big(L_2^k(\bm{m},\Lambda_2)\Big)_{\Delta_2^*,\geq 1},\\
&C_k^{(1)}(\bm{m},\Lambda_1)=B_k^{(1)}(\bm{m}+\bm{e},\Lambda_1),\quad C_k^{(2)}(\bm{m},\Lambda_2)=\Delta_2^{*}\cdot B_2^k(\bm{m}+\bm{e},\Lambda_2)\cdot\Delta_2^{*-1}.
\end{align*}
\begin{proposition}
The system of \eqref{particalLb} \eqref{Hl1L2} and \eqref{particalH} is well defined.
\end{proposition}
\begin{proof}
Firstly, let us show both sides of \eqref{particalLb} have the same forms. In fact, notice that $[B_k^{(1)}(\bm{m},\Lambda_1),L_1(\bm{m},\Lambda_1)]=-[(L_1^k(\bm{m},\Lambda_1))_{1,<0},L_1(\bm{m},\Lambda_1)]$
has the highest order $0$ with respect to $\Lambda_1$, which is consistent with the case of $\partial_{t_k^{(1)}}L_1(\bm{m},\Lambda_1)$. Similarly, 
  $[B_k^{(2)}(\bm{m},\Lambda_1),L_2(\bm{m},\Lambda_2)]=-[(L_2^k(\bm{m},\Lambda_2))_{\Delta_2^*,\leq0},L_2(\bm{m},\Lambda_2)]$
has the lowest order $-1$ with respect to $\Lambda_2$. Thus $\partial_{t_k^{(a)}}L_a=[B_k^{(a)},L_a]$ is well defined; that is $\partial_{t_k^{(a)}}L_a$ and $[B_k^{(a)},L_a]$ have the same expansions of $\Lambda_a$. As for $\partial_{t_k^{(a)}}L_b=[\pi_b(B_k^{(a)}),L_b]$, we can find that it is still well defined, from the following facts. When $k>0$, we can find $\pi_1(\Lambda_2^{-k})-1$ has the highest order $-1$ with respect to $\Lambda_1$, while $\pi_2(\Lambda_1^k)$ has the lowest order $0$ with respect to $\Lambda_2$. 

Next let us show the right hand side of \eqref{particalH} is a function, thus 
\eqref{particalH} is well defined. Actually by $H(\bm{m})L_1(\bm{m})=L_1(\bm{m}+\bm{e})H(\bm{m})$, we know $H(\bm{m})L_1^k(\bm{m})=L_1^k(\bm{m}+\bm{e})H(\bm{m})$, which implies 
\begin{align}
  B_k^{(1)}(\bm{m}+\bm{e})H(\bm{m})-H(\bm{m})B_k^{(1)}(\bm{m})=H(\bm{m})(L_1^k(\bm{m}))_{1,<0}-(L_1^k(\bm{m}+\bm{e}))_{1,<0}H(\bm{m})\label{HL1k-L1kH}
\end{align}
If assume \eqref{HL1k-L1kH}$=a(\bm{m},\Lambda_1)\Lambda_1\Lambda_2+b(\bm{m},\Lambda_1)$, then $b(\bm{m},\Lambda_1)$ has the lowest order $0$ and the highest order $0$, while $a(\bm{m},\Lambda_1)$ has the lowest order $0$ and the highest order $-1$. Thus $b(\bm{m},\Lambda_1)$ is a function and $a(\bm{m},\Lambda_1)=0$, which implies that $\partial_{t_k^{(1)}}H(\bm{m})=B_k^{(1)}(\bm{m}+\bm{e})H(\bm{m})-H(\bm{m})B_k^{(1)}(\bm{m})$ is well defined. Similarly, we can prove the case of $\partial_{t_k^{(2)}}H(\bm{m})$.

Finally, let us explain that \eqref{Hl1L2} is consistent with \eqref{particalLb} and \eqref{particalH}. If insert \eqref{particalLb}--\eqref{particalH} into $\partial_{t_k^{(1)}}\Big(H(\bm{m})L_1(\bm{m},\Lambda_1)-L_1(\bm{m}+\bm{e},\Lambda_1)H(\bm{m})\Big)$, it will become zero. For $\partial_{t_k^{(2)}}\Big(H(\bm{m})L_1(\bm{m},\Lambda_1)-L_1(\bm{m}+\bm{e},\Lambda_1)H(\bm{m})\Big)$, let us denote it to be $A_k(\bm{m},\La_1,\La_2)$, then we can get the following relation by \eqref{particalLb} and \eqref{particalH}, 
\begin{align*}
A_k(\bm{m},\La_1,\La_2)=D_k(\bm{m},\La_1,\La_2)L_1(\bm{m})-L_1(\bm{m}+\bm{e})D_k(\bm{m},\La_1,\La_2),
\end{align*}
where $D_k(\bm{m},\La_1,\La_2)=\partial_{t_k^{(2)}}H(\bm{m})+H(\bm{m})\pi_1(B_k^{(2)}(\bm{m}))-\pi_1(B_k^{(2)}(\bm{m}+\bm{e}))H(\bm{m})$. 
Thus if we can show $D_k=0$, then $A_k=0$. Notice that the coefficient of $\La_2$ in $D_k$ is zero, thus $D_k\in\E_{(2)}^0$.
Further by \eqref{particalH}, $$D_k(\bm{m})=C_k^{(2)}(\bm{m})H(\bm{m})-H(\bm{m})\Big(B_k^{(2)}(\bm{m})-\pi_1(B_k^{(2)}(\bm{m}))\Big)-\pi_1(B_k^{(2)}(\bm{m}+\bm{e}))H(\bm{m}).$$ So by $B_k^{(2)}(\bm{m})-\pi_1(B_k^{(2)}(\bm{m}))\in\E_{(1)}H$, we can get $D_k\in\E_{(1)}H$. Thus $D_k=0$, since $\E_{(2)}^0\cap\E_{(1)}H=\{0\}$.
Similarly, we can prove $\partial_{t_k^{(a)}}\Big(H(\bm{m})L_2(\bm{m})-\Delta_2^*L_2(\bm{m}+\bm{e})\Delta_2^{*-1}H(\bm{m})\Big)=0$ after inserting \eqref{particalLb}--\eqref{particalH}.
\end{proof}

\begin{lemma}\label{pi(B),L=pi(B,L)}
For $a,b=1,2$ and $k,l>0$,
\begin{align}\label{eq:pi(B),L=pi(B,L)}
  [\pi_a(B_k^{(b)}),L_a^l]=\pi_a([B_k^{(b)},L_a^l]).
\end{align}
\end{lemma}
\begin{proof}
It is obviously that \eqref{eq:pi(B),L=pi(B,L)} is correct when $a=b$. As for $a\neq b$, let us denote $A_{k,l}^{a,b}=[\pi_a(B_k^{(b)}),L_a^l]-\pi_a([B_k^{(b)},L_a^l])$, then $A_{k,l}^{a,b}\in \E_{(a)}^0$. If assume $\pi_a(B_k^{(b)})=B_k^{(b)}+D_k^{a,b}H$ and $\pi_a([B_k^{(b)},L_a^l])=[B_k^{(b)},L_a^l]+E_{k,l}^{a,b}H$, then $$A_{k,l}^{a,b}=D_k^{a,b}HL_a^l+(E_{k,l}^{a,b}-L_a^lD_k^{a,b})H.$$ So by \eqref{Hl1L2}, we can know $A_{k,l}^{a,b}\in\E_{(a)}^0\cap\E_{(a)}H=\{0\}$, which means $A_{k,l}^{a,b}=0$.
\end{proof}
\begin{lemma}\label{bk_1 and bl_2}$B_k^{(a)}$ satisfies the following relations,
$$\partial_{t_k^{(2)}}B_l^{(1)}=\pi_1([B_k^{(2)},B_l^{(1)}])_{1,\geq0},\quad \partial_{t_k^{(1)}}B_l^{(2)}=\pi_2([B_k^{(1)},B_l ^{(2)}])_{\Delta_2^*,\geq1}.$$
\end{lemma}
\begin{proof}
Firstly by \eqref{particalLb} and Lemma \ref{pi(B),L=pi(B,L)}, we know
  $$\partial_{t_k^{(2)}}B_l^{(1)}=[\pi_1(B_k^{(2)}),L_1^{l}]_{1,\geq0}=\pi_1([B_k^{(2)},L_1^{l}])_{1,\geq0}.$$
Next for $i,j\geq 0$, there exists $A_j\in \E_{(1)}$, such that $$\La_1^{-i-1}\La_2^{-j}=\La_1^{-i-1}(\pi_1(\La_2^{-j})+A_jH)=\La_1^{-i-1}\pi_1(\La_2^{-j})+\La_1^{-i-1}A_jH,$$ which means that $\pi_1(\La_1^{-i-1}\La_2^{-j})=\La_1^{-i-1}\pi_1(\La_2^{-j})$. So by \eqref{lemmapi1pi2}, we find the highest order of $\pi_1(\La_1^{-i-1}\La_2^{-j})$ is $-i-1$. Thus $\pi_1([B_k^{(2)},(L_1^{l})_{1,\leq-1}])_{1,\geq0}=0$, and $\partial_{t_k^{(2)}}B_l^{(1)}=\pi_1([B_k^{(2)},B_l^{(1)}])_{1,\geq0}$. 

Similarly by $\partial_{t_k^{(1)}}B_l^{(2)}=\pi_2([B_k^{(1)},L_2^l])_{\Delta_2^*,\geq1}$ and $(\pi_2(\La_1^i\Delta_2^{*-j}))_{\Delta_2^*,\geq1}=0 \ (i,j\geq0)$, we can prove $\partial_{t_k^{(1)}}B_l^{(2)}=\pi_2([B_k^{(1)},B_l ^{(2)}])_{\Delta_2^*,\geq1}$. 
\end{proof}

\begin{proposition}\label{bk_1,bl_2}
$B_k^{(a)}$ satisfies the following relations,
$$\partial_{t_k^{(a)}}B_l^{(b)}-\partial_{t_l^{(b)}}B_k^{(a)}+[B_l^{(b)},B_k^{(a)}]\in\E H,$$
with $\E=\mathcal{B}[\La_1,\La_2,\La_1^{-1},\La_2^{-1}].$
\end{proposition}
\begin{proof}
If denote $D_{k,l}^{(a,b)}=\partial_{t_k^{(a)}}B_l^{(b)}-\partial_{t_l^{(b)}}B_k^{(a)}+[B_l^{(b)},B_k^{(a)}]$, then $a=b$, $D_{k,l}^{(a,b)}=0$, which can be proved by direct computation, e.g.
\begin{align*}
      &\partial_{t_k^{(1)}}B_l^{(1)}-\partial_{t_k^{(1)}}B_k^{(1)}+[B_l^{(1)},B_k^{(1)}]\\
      &=[B_k^{(1)},L_1^l]_{1,\geq0}-[B_l^{(1)},L_1^k]_{1,\geq0}+[B_l^{(1)},B_k^{(1)}]\\
      &=[B_k^{(1)},B_l^{(1)}]_{1,\geq0}+[B_k^{(1)},(L_1^l)_{1,<0}]_{1,\geq0}-[B_l^{(1)},L_1^k]_{1,\geq0}+[B_l^{(1)},B_k^{(1)}]\\
      &=[B_k^{(1)},(L_1^l)_{1,<0}]_{1,\geq0}-[B_l^{(1)},L_1^k]_{1,\geq0}\\
      &=[L_1^k-(L_1^k)_{1,<0},(L_1^l)_{1,<0}]_{1,\geq0}+[L_1^k,B_l^{(1)}]_{1,\geq0}\\
      &=[L_1^k,(L_1^l)_{1,<0}+B_l^{(1)}]_{1,\geq0}=[L_1^k,L_1^l]_{1,\geq0}=0.
    \end{align*}

As for $a\neq b$, we only need to show $D_{k,l}^{(1,2)}\in\E H$, since $D_{k,l}^{(a,b)}=-D_{l,k}^{(b,a)}$. In fact by using $\La_1\Delta_2^*=\rho(\bm{m}-\bm{e}_2)\Delta_2^{*}+\rho(\bm{m}-\bm{e}_2)-\Lambda_2^{-1}H(\bm{m})$, we can know by induction on $i,j>0$ that $$\La_1^j\Delta_2^{*j}=\pi_1(\La_1^j\Delta_2^{*j})_{1,\geq0}+\pi_2(\La_1^j\Delta_2^{*j})_{\Delta_2^*,\geq1}+A_{ij}H,$$
where $A_{ij}\in\E$. Thus we have $$[B_k^{(1)},B_l^{(2)}]-\pi_1([B_k^{(1)},B_l^{(2)}])_{1,\geq0}-\pi_2([B_k^{(1)},B_l^{(2)}])_{\Delta_2^*,\geq1}\in\E H.$$ Further by Lemma \ref{bk_1 and bl_2}, we can finally obtain $D_{k,l}^{(1,2)}\in\E H$
\end{proof}
\begin{corollary}\label{abc=12}
For $a,b,c=1,2,$
\begin{align}\label{eqabc=12}
 \partial_{t_k^{(a)}}\pi_b(B_l^{(c)})-\partial_{t_l^{(c)}}\pi_b(B_k^{(a)})+[\pi_b(B_l^{(c)}),\pi_b(B_k^{(a)})]=0.
\end{align}
Thus $[\partial_{t_k^{(a)}},\partial_{t_l^{(c)}}]=0$ on $L_i$ and $H$.
\end{corollary}
\begin{proof}
If denote the left hand side of \eqref{eqabc=12} to be $A_{k,l}^{(a,b,c)}$, then $A_{k,l}^{(a,b,c)}\in\E_{(b)}^0$. On the other hand, if assume $\pi_b(B_l^{(c)})=B_l^{(c)}+C_l^{(b,c)}H$, with $C_l^{(b,c)}\in\E_{(b)}$, then 
\begin{align*}
  A_{k,l}^{(a,b,c)}=&\partial_{t_k^{(a)}}B_l^{(c)}-\partial_{t_l^{(c)}}B_k^{(a)}+[B_l^{(c)},B_k^{(a)}]\\
  &+(\partial_{t_k^{(a)}}C_l^{(b,c)}-\partial_{t_l^{(c)}}C_k^{(b,c)}+B_l^{(c)}C_k^{(b,a)}-B_k^{(a)}C_l^{(b,c)}+C_l^{(b,c)}HC_k^{(b,a)}-C_k^{(b,a)}HC_l^{(b,c)})H\\
  &+C_l^{(b,c)}(\partial_{t_k^{(a)}}H+HB_k^{(a)})-C_k^{(b,a)}(\partial_{t_l^{(c)}}H+HB_l^{(c)}).
\end{align*}
Further by \eqref{particalH} and Proposition \ref{bk_1,bl_2}, we can know $A_{k,l}^{(a,b,c)}\in\E_bH$. So finally $A_{k,l}^{(a,b,c)}\in\E_bH\bigcap\E_{(b)}^0=\{0\}$, which means $A_{k,l}^{(a,b,c)}=0$.

As for $[\partial_{t_k^{(a)}},\partial_{t_l^{(c)}}]=0$, it can be proved directly by \eqref{particalLb} \eqref{particalH} and \eqref{eqabc=12}.
\end{proof}

\begin{proposition}\label{proexsit of wave operates}
Given 2--Toda Lax operators $L_1(\Lambda_1)=\La_1+\sum_{i=0}^{+\infty}u_i^{(1)}\Lambda_1^{-i}$, $L_2(\Lambda_2)=u_{-1}^{(2)}\La_2^{-1}+\sum_{i=0}^{+\infty}u_i^{(2)}\La_2^i$ and the operator $H=\La_1\Delta_2+\rho$ satisfying \eqref{particalLb}--\eqref{particalH}, there exist wave operators $S_1=1+\sum_{k=1}^{+\infty}a_k^{(1)}\Lambda_1^{-k}$ and $S_2=a_0^{(2)}+\sum_{k=1}^{+\infty}a_k^{(2)}\Lambda_2^{-k} \ (a_0^{(2)}\neq0)$ such that
\begin{align*}
&L_1=S_1\Lambda_1S_1^{-1}, \quad L_2=S_2\Lambda_2^{-1}S_2^{-1},\quad H=-\Lambda_1\Lambda_2S_1\Delta_2^{*}S_1^{-1}=-\Lambda_1\Delta_2S_2\Delta_1^{*}S_2^{-1},\\
&\partial_{t_k^{(1)}}S_1=B_k^{(1)}S_1-S_1\Lambda_1^k,\quad \partial_{t_k^{(1)}}S_2
=\left(B_k^{(1)}S_1\Delta_2^{*-1}S_1^{-1}\right)_{1,[0]}\Delta_2^*S_2,\\
 &\partial_{t_k^{(2)}}S_1
=\left(B_k^{(2)}S_2\Delta_1^{-1}
S_2^{-1}\Delta_2^{*-1}\right)_{2,[0]}S_1,\quad \partial_{t_k^{(2)}}S_2=B_k^{(2)}S_2-S_2\Lambda_2^{-k}.
\end{align*}
\end{proposition}
\begin{proof}
Firstly it is obviously that there exist $\overline{S}_1=1+\sum_{i=1}^{+\infty}b_i^{(1)}\La_1^{-i}$ and $\overline{S}_2=b_0^{(2)}+\sum_{i=1}^{+\infty}b_i^{(2)}\La_2^i \ (b_0^{(2)}\neq0) $ such that $L_1=\overline{S}_1\Lambda_1\overline{S}_1^{-1}$ and  $L_2=\overline{S}_2\Lambda_2^{-1}\overline{S}_2^{-1}$.

Next consider the following system of $S_1$ and $S_2$ 
\begin{equation}\label{particalS1S2L12}
\begin{cases}
\partial_{t_k^{(1)}}S_1=B_k^{(1)}S_1-S_1\Lambda_1^k,\quad \partial_{t_k^{(1)}}S_2=\pi_2(B_k^{(1)})S_2,\\
\partial_{t_k^{(2)}}S_1=\pi_1(B_k^{(2)})S_1,\quad \partial_{t_k^{(2)}}S_2=B_k^{(2)}S_2-S_2\Lambda_2^{-k},\\
\La_2(S_1)=(\La_2-\La_1^{-1}H)S_1,\quad \La_1(S_2)=(\La_1+\La_2^{-1}\Delta_2^{*-1}H)S_2,\\
S_1|_{t=0}=\widetilde{S}_1|_{t=0}, \quad S_2|_{t=0}=\widetilde{S}_2|_{t=0},
\end{cases}
\end{equation}
where $S_a$ has the form of \eqref{SaandwideSa}, $B_k^{(1)}=(L_1^k)_{1,\geq 0},$ $B_k^{(2)}=(L_2^k)_{\Delta_2^*,\geq 1}$ and $\pi_a$ is the projection $\E_{(a)}=\E_{(a)}^0\oplus\E_{(a)}H\to\E_{(a)}^0$.
By \eqref{particalH} and \eqref{eqabc=12}, we can find $[\partial_{t_k^{(a)}},\partial_{t_l^{(b)}}]S_c=0$ and $\partial_{t_k^{(a)}}(\La_b(S_c))=\La_b(\partial_{t_k^{(a)}}(S_c))$. Thus, we can know the system \eqref{particalS1S2L12} has a unique solution $\widehat{S}_1=1+\sum_{i=1}^{+\infty}\widehat{b}_i^{(1)}\La_1^{-i}$ and $\widehat{S}_2=\widehat{b}_0^{(2)}+\sum_{i=1}^{+\infty}\widehat{b}_i^{(2)}\La_2^i \ (\widehat{b}_0^{(2)}\neq0)$. 

If denote $\widehat{W}_1=L_1\widehat{S}_1-\widehat{S}_1\La_1$ and $\widehat{W}_2=L_2\widehat{S}_2-\widehat{S}_2\La_2^{-1}$, then we find $\widehat{W}_a\ (a=1,2)$ satisfies 
\begin{align}\label{particalW=BW-WLa}
  \partial_{t_k^{(a)}}\widehat{W}_a=B_k^{(a)}\widehat{W}_a-\widehat{W}_a\Lambda_a^{(3-2a)k},\quad \partial_{t_k^{(a)}}\widehat{W}_{3-a}=\pi_{3-a}(B_k^{(a)})\widehat{W}_{3-a},\quad a=1,2.
\end{align}
By using \eqref{Hl1L2}, we can know
\begin{align*}
  &\La_2(L_1)=(\La_2-\La_1^{-1}H)L_1\cdot\iota_{\La_1^{-1}}(\La_2-\La_1^{-1}H)^{-1},\\
  &\La_1(L_2)=(\La_1+\La_2^{-1}\Delta_2^{*-1}H)L_2\cdot\iota_{\La_2}(\La_1+\La_2^{-1}\Delta_2^{*-1}H)^{-1}.
\end{align*}
So based upon these relations, we can show that
\begin{align}\label{La1WLa2W}
  \La_2(\widehat{W}_1)=(\La_2-\La_1^{-1}H)\widehat{W}_1,\quad \La_1(\widehat{W}_2)=(\La_1+\La_2^{-1}\Delta_2^{*-1}H)\widehat{W}_2.
\end{align}
Therefore by similar reason as \eqref{particalS1S2L12}, there exists one unique solution for the following system of $W_1=\sum_{i=1}^{+\infty}w_i^{(1)}\La_1^{-i}$ and $W_2=\sum_{i=0}^{+\infty}w_i^{(2)}\La_2^i$ ($b_0^{(2)}$ can be zero)
\begin{equation}\label{particalV1V2La12}
\begin{cases}
 \partial_{t_k^{(a)}}W_a=B_k^{(a)}W_a-W_a\Lambda_a^{(3-2a)k},\quad \partial_{t_k^{(a)}}W_{3-a}=\pi_{3-a}(B_k^{(a)})W_{3-a},\quad a=1,2,\\
 \La_2(W_1)=(\La_2-\La_1^{-1}H)W_1,\quad \La_1(W_2)=(\La_1+\La_2^{-1}\Delta_2^{*-1}H)W_2,\\
  W_1|_{t=0}=W_2|_{t=0}=0.
\end{cases}
\end{equation}
Obviously by \eqref{particalW=BW-WLa} \eqref{La1WLa2W} and  $\widehat{W}_a|_{t=0}=L_a|_{t=0}\overline{S}_a|_{t=0}-
\overline{S}_a|_{t=0}\La_a^{3-2a}=0$ ($a=1,2$), $\widehat{W}$ is the solution, while $W=0$ is also another solution. By uniqueness of the solution for \eqref{particalV1V2La12}, we find $\widehat{W}=0$, which means that $\widehat{S}_1$ and $\widehat{S}_2$ are the required wave operators in Proposition \ref{proexsit of wave operates}.
\end{proof}
\begin{remark}
$S_1$ and $S_2$ can be up to the multiplications of $C_1=\sum_{j=0}^{\infty}c_1^{(j)}\Lambda_1^{-j}$ and $C_2=\sum_{j=0}^{\infty}c_2^{(j)}\Lambda_2^{j}$ on the right sides respectively. Here $c_i^{(j)}$ is some constant without depending on $\bm{m}$ and $t$.
\end{remark}

\begin{corollary}\label{coro exist oftaufunction}
Given the wave operators $S_1$ and $S_2$ in Proposition \ref{proexsit of wave operates}, define the wave functions $\Psi_a$ and the adjoint wave functions $\widetilde{\Psi}_a$ in the way below,
\begin{align*}
&\Psi_1(\bm{m},t,z)=e^{\xi(t^{(1)},z)}S_1(\bm{m},t,\La_1)(z^{m_1}),\\
&\Psi_2(\bm{m},t,z)=e^{\xi(t^{(2)},z^{-1})}S_2(\bm{m},t,\La_2)(z^{m_2}),\\
&\widetilde{\Psi}_1(\bm{m},t,z)=ze^{-\xi(t^{(1)},z)}(S_1^{-1}(\bm{m}-\bm{e}_1,t,\La_1))^*(z^{-m_{1}}),\\
&\widetilde{\Psi}_2(\bm{m},t,z)=e^{-\xi(t^{(2)},z^{-1})}\Delta_2^{-1}(S_2^{-1}(\bm{m}-\bm{e}_1,t,\La_2))^*(z^{-m_{2}}).
\end{align*}
then $\Psi_a$ and $\widetilde{\Psi}_a$ satisfy the following relations,
\begin{align*}
&\partial_{t_k^{(1)}}\Psi_a(\bm{m})=B_k^{(1)}(\bm{m},\Lambda_1)\Big(\Psi_a(\bm{m})\Big),
\quad \partial_{t_k^{(2)}}\Psi_a(\bm{m})=B_k^{(2)}(\bm{m},\Lambda_2)\Big(\Psi_a(\bm{m})\Big),\\
&\partial_{t_k^{(1)}}\widetilde{\Psi}_a(\bm{m})=-B_k^{*(1)}(\bm{m}-\bm{e}_1,\Lambda_1)\left(\widetilde{\Psi}_a(\bm{m})\right),\quad \partial_{t_k^{(2)}}\widetilde{\Psi}_a(\bm{m})=-\Delta_2^{-1}B_k^{*(2)}(\bm{m}-\bm{e}_1,\Lambda_2)\Delta_2\left(\widetilde{\Psi}_a(\bm{m})\right),\\
&H(\bm{m})(\Psi_a(\bm{m})) = 0,\quad \widetilde{H}(\bm{m})(\widetilde{\Psi_a}(\bm{m}+\bm{e}_1)) = 0,
\end{align*}
where $\widetilde{H}(\bm{m})=H^*(\bm{m}-\bm{e})$.
\end{corollary}
\begin{proof}
Notice that $\partial_{t_k^{(a)}}\Psi_a$, $\partial_{t_k^{(a)}}\widetilde{\Psi}_a$, $H(\Psi_a)$ and $\widetilde{H}(\widetilde{\Psi}_a)$ can be computed directly by Proposition \ref{proexsit of wave operates}. As for $\partial_{t_k^{(a)}}\Psi_{3-a}$, $\partial_{t_k^{(a)}}\widetilde{\Psi}_{3-a}$, they can be computed by $\partial_{t_k^{(a)}}S_{3-a}=\pi_{3-a}(B_k^{(a)})S_{3-a}$ and $H(\Psi_a)=\widetilde{H}(\widetilde{\Psi}_a)=0$.
\end{proof}

\begin{proposition}\label{pro dressing to bilinear}
The wave functions $\Psi_a$ and the adjoint wave functions $\widetilde{\Psi}_a$ defined in Corollary \ref{coro exist oftaufunction} satisfy the following bilinear equation
\begin{align}\label{dressing to bilinear}
\oint_{C_R}\frac{dz}{2\pi iz}\Psi_{1}(\bm{m},t,z)\widetilde{\Psi}_{1}(\bm{m}',t',z)=\oint_{C_r}\frac{dz}{2\pi iz}\Psi_{2}(\bm{m},t,z)\widetilde{\Psi}_{2}(\bm{m}',t',z), \quad m_1-m_2\geq m_1'-m_2'.
\end{align}
\begin{proof}
Firstly by Corollary \ref{coro exist oftaufunction} and the Taylor expansions at $t'=t$, we only need to show 
\begin{align*}
\oint_{C_R}\frac{dz}{2\pi iz}\Psi_{1}(\bm{m},t,z)\widetilde{\Psi}_{1}(m'_1-k,m'_2+l,t,z)=\oint_{C_r}\frac{dz}{2\pi iz}\Psi_{2}(\bm{m},t,z)\widetilde{\Psi}_{2}(m'_1-k,m'_2+l,t,z),
\end{align*}
for any $m_1-m_2\geq m'_1-m'_2$, $k,l\geq0$, which is further equivalent to
\begin{align}\label{bilinear eq lose kl}
\oint_{C_R}\frac{dz}{2\pi iz}\Psi_{1}(\bm{m},t,z)\widetilde{\Psi}_{1}(\bm{m}',t,z)=\oint_{C_r}\frac{dz}{2\pi iz}\Psi_{2}(\bm{m},t,z)\widetilde{\Psi}_{2}(\bm{m}',t,z), \quad m_1-m_2\geq m'_1-m'_2.
\end{align}
By Lemma \ref{ABlemma}, \eqref{bilinear eq lose kl} is equivalent to
\begin{align*}
&\sum_{j\geq 1}\left(S_1(\bm{m},\Lambda_1)\Lambda_2^{j}S_1^{-1}(\bm{m},\Lambda_1)\Lambda_1\right)_{1,\leq j}=\sum_{j\leq 1}\left(S_2(\bm{m},\Lambda_2)\Lambda_1^{j}\Lambda_1^{-1}S_2^{-1}(\bm{m},\Lambda_2)(\Delta_2^{*})^{-1}\Lambda_1\right)_{2,\geq j},\\
\Leftrightarrow&\sum_{j\geq1}S_1(\bm{m},\Lambda_1)\La_2^jS_1^{-1}(\bm{m},\Lambda_1)=\sum_{j\leq1}S_2(\bm{m},\Lambda_2)\La_1^j\La_1^{-1}S_2^{-1}(\bm{m},\Lambda_2)\Delta_2^{*-1},\\
\Leftrightarrow&S_1(\bm{m},\Lambda_1)\Delta_2^{*-1}S_1^{-1}(\bm{m},\Lambda_1)=S_2(\bm{m},\Lambda_2)\Delta_1^{-1}\La_1S_2^{-1}(\bm{m},\Lambda_2)\Delta_2^{*-1},\\
\Leftrightarrow&S_1(\bm{m},\Lambda_1)\Delta_2^{*}S_1^{-1}(\bm{m},\Lambda_1)=-\Delta_2^{*}S_2(\bm{m},\Lambda_2)\Delta_1^{*}S_2^{-1}(\bm{m},\Lambda_2).
\end{align*} 
Notice that the last equation holds by $H=-\Lambda_1\Lambda_2S_1\Delta_2^{*}S_1^{-1}=-\Lambda_1\Delta_2S_2\Delta_1^{*}S_2^{-1}.$
\end{proof}
\end{proposition}

Now we will prove the existence of the tau function from the bilinear equation in Proposition \ref{pro dressing to bilinear} for the 2dKP hierarchy, which is given in the proposition below.
\begin{proposition}\label{the existence of tau function}
  Given the 2dKP wave function $\Psi_a$ and the 2dKP adjoint wave function $\widetilde{\Psi}_a$ satisfying the 2dKP bilinear equation \eqref{dressing to bilinear}, there exist one tau function $\tau_m(t)$ such that
  \begin{equation}\label{eq the existence of tau function}
  \begin{aligned}
&\Psi_{1}(\bm{m},t,z)=z^{m_1}e^{\xi(t^{(1)},z)}\frac{\tau_{\bm{m}}(t-[z^{-1}]_1)}{\tau_{\bm{m}}(t)},\\
&\widetilde{\Psi}_{1}(\bm{m},t,z)=z^{-m_1+1}e^{-\xi(t^{(1)},z)}\frac{\tau_{\bm{m}}(t+[z^{-1}]_1)}{\tau_{\bm{m}}(t)},\\
&\Psi_{2}(\bm{m},t,z)=z^{m_2}e^{\xi(t^{(2)},z^{-1})}\frac{\tau_{\bm{m}+\bm{e}}(t-[z]_2)}{\tau_{\bm{m}}(t)},\\
&\widetilde{\Psi}_{2}(\bm{m},t,z)=z^{-m_2+1}e^{-\xi(t^{(2)},z^{-1})}\frac{\tau_{\bm{m}-\bm{e}}(t+[z]_2)}{\tau_{\bm{m}}(t)}.
\end{aligned}
\end{equation}
\end{proposition}

To prove this proposition, we need to do the following preparations. Firstly, let us denote $\psi_a(\bm{m},t,z)=z^{-m_a}e^{-\xi(t^{(a)},z^{3-2a})}\Psi_a(\bm{m},t,z)$ and $\widetilde{\psi}_a(\bm{m},t,z)=z^{m_a-1}e^{\xi(t^{(a)},z^{3-2a})}\widetilde{\Psi}_a(\bm{m},t,z) \ (a=1,2)$, then by Corollary \ref{coro exist oftaufunction} $\psi_a(\bm{m},t,z)$ and $\widetilde{\psi}_a(\bm{m},t,z)$ have the following expansions of $z$,
\begin{align*}
&\psi_1=1+\sum_{k=1}^{+\infty}a_k^{(1)}z^{-k}, \quad \psi_2=a_0^{(2)}+\sum_{k=1}^{+\infty}a_k^{(2)}z^k,\\
&\widetilde{\psi}_1=1+\sum_{k=1}^{+\infty}\widetilde{a}_k^{(1)}z^{-k}, \quad \widetilde{\psi}_2=\widetilde{a}_0^{(2)}+\sum_{k=1}^{+\infty}\widetilde{a}_k^{(2)}z^k.
\end{align*}
\begin{lemma}\label{le psi1psi2wide}
$\psi_a$ and $\widetilde{\psi}_a$ satisfy the following relations,
\begin{align}
&\psi_1(\bm{m},t,\lambda_1)\widetilde{\psi}_1(\bm{m},t-[\lambda_1^{-1}]_1-[\lambda_2^{-1}]_1,\lambda_1)=\psi_1(\bm{m},t,\lambda_2)\widetilde{\psi}_1(\bm{m},t-[\lambda_1^{-1}]_1-[\lambda_2^{-1}]_1,\lambda_2), \label{psi1,psi1}\\
&\psi_1(\bm{m},t,\lambda_1)\widetilde{\psi}_1(\bm{m}+\bm{e},t-[\lambda_1^{-1}]_1-[\lambda_2]_2,\lambda_1)=\psi_2(\bm{m},t,\lambda_2)\widetilde{\psi}_2(\bm{m}+\bm{e},t-[\lambda_1^{-1}]_1-[\lambda_2]_2,\lambda_2), \label{psi1,psi2}\\
&\psi_2(\bm{m},t,\lambda_1)\widetilde{\psi}_2(\bm{m}+2\bm{e},t-[\lambda_1]_2-[\lambda_2]_2,\lambda_1)=\psi_2(\bm{m},t,\lambda_2)\widetilde{\psi}_2(\bm{m}+2\bm{e},t-[\lambda_1]_2-[\lambda_2]_2,\lambda_2). \label{psi2,psi2}
\end{align}
\end{lemma}
\begin{proof}
Firstly in terms of $\psi_a$ and $\widetilde{\psi}_a$, the 2dKP bilinear equation will become
\begin{align}\label{CRCrpsi1psi2}
\oint_{C_R}\frac{dz}{2\pi i}z^{m_1-m'_1}\psi_{1}(\bm{m},t,z)\widetilde{\psi}_{1}(\bm{m}',t',z)e^{\xi(t^{(1)}-t^{(1)'},z)}=\oint_{C_r}\frac{dz}{2\pi i}z^{m_2-m'_2}\psi_{2}(\bm{m},t,z)\widetilde{\psi}_{2}(\bm{m}',t',z)e^{\xi(t^{(2)}-t^{(2)'},z)},  
\end{align}
where $m_1-m_2\geq m'_1-m'_2.$ Then \eqref{psi1,psi1}--\eqref{psi2,psi2} can be obtained by setting
\begin{itemize}
  \item $\bm{m}'=\bm{m}, \quad t'=t-[\lambda_1^{-1}]_1-[\lambda_2^{-1}]_1$,
  \item $\bm{m}'=\bm{m}+\bm{e}, \quad t'=t-[\lambda_1^{-1}]_1-[\lambda_2]_2$,
  \item $\bm{m}'=\bm{m}+2\bm{e}, \quad t'=t-[\lambda_1]_2-[\lambda_2]_2$.
\end{itemize}
\end{proof}
\begin{lemma}\label{psiandwidepsi=1}
In the 2dKP bilinear equation \eqref{CRCrpsi1psi2}, $\psi_a$ and $\widetilde{\psi}_a$ are related by
\begin{align}
  &\psi_1(\bm{m},t,\lambda)\widetilde{\psi}_1(\bm{m},t-[\lambda^{-1}]_1,\lambda)=1, \label{psi1widepsi1}\\
  &\psi_2(\bm{m},t,\lambda)\widetilde{\psi}_2(\bm{m}+\bm{e},t-[\lambda]_2,\lambda)=1. \label{psi2widepsi2}
\end{align}
\begin{proof}
If set $\lambda_1=\lambda$, $\lambda_2=\infty$, then \eqref{psi1,psi1} becomes \eqref{psi1widepsi1}, while \eqref{psi2widepsi2} can be obtained by setting $\lambda_1=\infty$, $\lambda_2=\lambda$ in \eqref{psi2,psi2}.
\end{proof}
\end{lemma}
By Lemma \ref{psiandwidepsi=1}, we can rewrite the relations in Lemma \ref{le psi1psi2wide} by expressing $\widetilde{\psi}_a$ in terms of $\psi_a$, which are given in the lemma below.
\begin{lemma}
$\psi_a$ satisfies 
\begin{align}
  &\psi_1(\bm{m},t,\lambda_1)\psi_1(\bm{m},t-[\lambda_1^{-1}]_1,\lambda_2)=\psi_2(\bm{m},t,\lambda_2)\psi_2(\bm{m},t-[\lambda_2^{-1}]_1,\lambda_1),\label{psi1psi1=psi2psi2}\\
  &\psi_1(\bm{m},t,\lambda_1)\psi_2(\bm{m},t-[\lambda_1^{-1}]_1,\lambda_2)=\psi_2(\bm{m},t,\lambda_2)\psi_1(\bm{m}+\bm{e},t-[\lambda_2]_2,\lambda_1),\label{psi1psi2=psi2psi1}\\
  &\psi_2(\bm{m},t,\lambda_1)\psi_2(\bm{m}+\bm{e},t-[\lambda_1]_2,\lambda_2)=\psi_1(\bm{m},t,\lambda_2)\psi_1(\bm{m}+\bm{e},t-[\lambda_2]_2,\lambda_1)\label{psi1psi1=psi2psi2}.
\end{align}
\end{lemma}
Further if set $\lambda_1=\lambda$, $\lambda_2=0$ in \eqref{psi1psi2=psi2psi1} and \eqref{psi1psi1=psi2psi2}, we can obtain the lemma below.
\begin{lemma}\label{exiloga=de12log}
If denote $\Delta_{12}=\Lambda_1\Lambda_2-1$, then 
\begin{align}
  &(e^{-\xi(\widetilde{\partial}_{t^{(1)}},\lambda^{-1})}-1)\log a_0^{(2)}(\bm{m},t)=\Delta_{12}\log \psi_1(\bm{m},t,\lambda),\label{exiloga=Delog}\\
  &(e^{-\xi(\widetilde{\partial}_{t^{(2)}},\lambda)}-1)\log a_0^{(2)}(\bm{m},t)=\Delta_{12}\log\left(\frac{\psi_2(\bm{m}-\bm{e},t,\lambda)}{a_0^{(2)}(\bm{m}-\bm{e},t)}\right).\label{exiloga=Deloga} 
\end{align}
\end{lemma}
\begin{lemma}
$\psi_a$ also satisfies
\begin{align} 
  &\psi_1(\bm{m},t,\la_1)\psi_2(\bm{m}-\bm{e},t-[\la_1^{-1}]_1,\la_2)a_0^{(2)}(\bm{m}-\bm{e},t)\nonumber\\
  =&\psi_2(\bm{m}-\bm{e},t,\la_2)\psi_1(\bm{m},t-[\la_2]_2,\la_1)a_0^{(2)}(\bm{m}-\bm{e},t-[\la_1^{-1}]_1),\label{psi1psi2a0=psi2psi1}\\
  &\psi_2(\bm{m},t,\la_1)\psi_2(\bm{m},t-[\la_1]_2,\la_2)a_0^{(2)}(\bm{m},t-[\la_2]_2)\nonumber\\
  =&\psi_2(\bm{m},t,\la_2)\psi_2(\bm{m},t-[\la_1]_2,\la_1)a_0^{(2)}(\bm{m},t-[\la_1]_2).\label{psi2psi2a0=psi2psi2}
\end{align}
\end{lemma}
\begin{proof}
\eqref{psi1psi2a0=psi2psi1} is equivalent to \eqref{psi1psi2=psi2psi1} after inserting $$a_0^{(2)}(\bm{m}-\bm{e},t-[\la_1^{-1}]_1)=\frac{\psi_1(\bm{m},t,\la_1)}{\psi_1(\bm{m}-\bm{e},t,\la_1)}a_0^{(2)}(\bm{m}-\bm{e},t)$$ from \eqref{exiloga=Delog}.
As for \eqref{psi2psi2a0=psi2psi2}, it can be derived from \eqref{psi1psi1=psi2psi2} and \eqref{exiloga=Deloga}. In fact by \eqref{exiloga=Deloga}, we know
\begin{align}\label{a02=psi2a02}
a_0^{(2)}(\bm{m},t-[\la]_2)=\frac{\psi_2(\bm{m},t,\la)}{\psi_2(\bm{m}-\bm{e},t,\la)}a_0^{(2)}(\bm{m}-\bm{e},t).
\end{align}
If substitute \eqref{a02=psi2a02} into \eqref{psi2psi2a0=psi2psi2}, we can find \eqref{psi2psi2a0=psi2psi2} is just \eqref{psi1psi1=psi2psi2}.
\end{proof}
After the preparations above, now let us see \textbf{the proof of Proposition \ref{the existence of tau function}}.

Firstly by Lemma \ref{psiandwidepsi=1}, we can find that the results of $\psi_a$ imply the cases of $\widetilde{\psi}_a$. Thus here we can only discuss the case of $\psi_a$. To do this, let us introduce $b_j^{(1)}$ and $b_l^{(2)} \ (j\geq1, l\geq0)$ in the way below
\begin{align*}
  &\log\psi_1(\bm{m},t,\la)=\sum_{j=1}^{+\infty}b_j^{(1)}(\bm{m},t)\la^{-j},\\
  &\log\psi_2(\bm{m},t,\la)=\sum_{l=0}^{+\infty}b_l^{(2)}(\bm{m},t)\la^{l},
\end{align*}
where we notice that $b_0^{(2)}(\bm{m})=\log a_0^{(2)}$. Then the relations \eqref{eq the existence of tau function} between $\psi_a$ and $\tau_m$ can be written in the forms below
\begin{equation}\label{pjbjde12}
\begin{cases}
 p_j(-\widetilde{\partial}_{t^{(1)}})\log\tau_m=b_j^{(1)}(\bm{m}),\\
 p_j(-\widetilde{\partial}_{t^{(2)}})\log\tau_m=b_j^{(2)}(\bm{m}-\bm{e}),\\
 \Delta_{12}\log\tau_m=b_0^{(2)}(\bm{m}), \quad j\geq1, 
\end{cases}
\end{equation}
where $p_j(x)$ is determined by $\exp(\xi(x,\la))=\sum_{j\geq0}^{+\infty}p_j(x)\la^j$, with $x=(x_1,x_2,\cdots)$. 

Notice that if the following relations hold
\begin{align}
  &p_j(-\widetilde{\partial}_{t^{(1)}})b_l^{(1)}(\bm{m})=p_l(-\widetilde{\partial}_{t^{(1)}})b_j^{(1)}(\bm{m}),\label{pjbj=plbj}\\
  &p_j(-\widetilde{\partial}_{t^{(1)}})b_l^{(2)}(\bm{m}-\bm{e})=p_l(-\widetilde{\partial}_{t^{(2)}})b_j^{(1)}(\bm{m}),\label{pjbj=plbj2}\\
  &p_j(-\widetilde{\partial}_{t^{(2)}})b_l^{(2)}(\bm{m}-\bm{e})=p_l(-\widetilde{\partial}_{t^{(2)}})b_j^{(2)}(\bm{m}-\bm{e}),\label{pjbj=plbj3}\\
  &p_j(-\widetilde{\partial}_{t^{(1)}})b_0^{(2)}(\bm{m})=\Delta_{12}b_j^{(1)}(\bm{m}),\label{pjb0=debj}\\
  &p_j(-\widetilde{\partial}_{t^{(2)}})b_0^{(2)}(\bm{m})=\Delta_{12}b_j^{(2)}(\bm{m}-\bm{e}), \quad j,l\geq1, \label{pjb0=debj2}
\end{align}
then \eqref{pjbjde12} has the solution $\log\tau_m$. In fact \eqref{pjbj=plbj} can be derived by comparing the coefficient of $\la_1^{-j}\la_2^{-l}$ in \eqref{psi1psi1=psi2psi2}. \eqref{pjbj=plbj2} and \eqref{pjbj=plbj3} can be obtained by \eqref{psi1psi2a0=psi2psi1} and \eqref{psi2psi2a0=psi2psi2}. Finally \eqref{pjb0=debj} and \eqref{pjb0=debj2} are just equivalent to the relations in Lemma \ref{exiloga=de12log}.

\section{reduction of the 2dKP hierarchy}
In this section, we will construct the reduction of the 2dKP hierarchy corresponding to the loop algebra $\widehat{sl}_{M+N}=sl_{M+N}[\la,\la^{-1}]\oplus\mathbb{C}c \ (M,N\geq1)$.

Firstly recall $sl_{M+N}[\la,\la^{-1}]$ is the set of the traceless $(M+N)\times(M+N)$ matrices with the matrix entries being the Laurent polynomials of $\la_i$ and the corresponding Lie bracket \cite{Kac2023,Kac1990} is given by $$[A\la^m,B\la^n]=(AB-BA)\la^{m+n}+m\delta_{m+n,0}c,\quad A,B\in sl_{M+N}.$$ If denote $e_{ij}=(\delta_{ia}\delta_{jb})_{1\leq a,b\leq M+N}$, then we can define the imbedding \cite{Kac2023,Kac1990} 
\begin{align*}
 i:\widehat{sl}_{M+N}&\rightarrow a_\infty=\overline{gl}_\infty\oplus\mathbb{C}c,\\
e_{ij}\la_n&\mapsto i(e_{ij}\la_n)=\sum_{l\in \mathbb{Z}}E_{i+(M+N)(l-n),j+(M+N)l},
\end{align*}
where $E_{ij}=(\delta_{ip}\delta_{jq}), \ p,q\in\mathbb{Z}$ and
$\overline{gl}_\infty=\left\{\sum_{k,l\in\mathbb{Z}}a_{ij}E_{kl}\ \biggr|\ a_{ij}=0,\ |k-l|\gg0\right\}.$
The image of $i$ is given by
\begin{align*}
a_{\infty}^{(M+N)}=\left\{A\in \overline{gl}_\infty\ \biggr|\ A_{p+(M+N)r,q+(M+N)r}=A_{pq},\ \sum_{i=1}^{M+N}A_{i-(M+N)k,i}=0,\ \forall k\in \mathbb{Z}\right\}.
\end{align*}
 Recall that the fermionic representation of $\mathcal{F}$ is given by \cite{Jimbo1983,Kac2023,Kac1990}
\begin{align*}
  \pi: &a_\infty\rightarrow End\mathcal{F}\\
       &E_{ij}\mapsto:\psi_i\psi_j:, \quad c\mapsto1.
\end{align*}
Thus $\rho=\pi\circ i$ can give the fermionic representation of $sl_{M+N}[\la,\la^{-1}]$. 

In terms of charged free fermions, the elements of $sl_{M+N}[\la,\la^{-1}]$ can be realized in the way below.
For any $a\in sl_n(\mathbb{C}[\la,\la^{-1}])$, there exists a unique $\rho(a)\in End\mathcal{F}$ having the following form,
$$\rho(a)=\sum_{i,j\in\mathbb{Z}} A_{ij}:\psi^+_{-i+1/2}\psi^-_{j-1/2}:, \quad A_{ij}=A_{i+(M+N)l,j+(M+N)l}.$$ 
So if introduce $$S_{M+N}=\sum\psi_j^+\otimes\psi_{j+M+N}^-,$$
then for $a\in sl_{M+N}[\la,\la^{-1}],$ $$[S_{M+N},1\otimes \rho(a)+\rho(a)\otimes1]=0.$$ Therefore if denote $\tau_l=\exp(a)|l\rangle$, where \begin{align*}
|l\rangle=\begin{cases}
\psi^{-}_{\frac{1}{2}+l}\psi^{-}_{\frac{3}{2}+l}\cdots\psi^{-}_{-\frac{1}{2}}|0\rangle,&\quad l<0;\\
|0\rangle,&\quad l=0;\\
\psi^{+}_{\frac{1}{2}-l}\psi^{+}_{\frac{3}{2}-l}\cdots\psi^{+}_{-\frac{1}{2}}|0\rangle,&\quad l>0,
\end{cases}
\end{align*} 
then we have 
\begin{align}\label{SMNtaul}
  S_{M+N}(\tau_l\otimes\tau_{l'})=0, \quad l\geq l',
\end{align}
where we have used $S_{M+N}(|l\rangle\otimes|l'\rangle)=0$. \eqref{SMNtaul} is just the fermionic version of the reduction of the dKP hierarchy corresponding to $\widehat{sl}_{M+N}$. 
The complete description of the $\widehat{sl}_{M+N}$--reduction of the dKP hierarchy contains \eqref{SMNtaul} and $S(\tau_l\otimes\tau_{l'})=0 \ (l\geq l').$ 

If we introduce $\psi_j^{\pm(a)}$ in the way below \cite{Kac2023,Kac2003}
\begin{align*}
&\psi_{Mi+p+1/2}^{+(1)}=\psi^+_{(M+N)i+p+1/2}, \ \quad \psi_{Mi-p-1/2}^{-(1)}=\psi^-_{(M+N)i-p-1/2}, \quad 0\leq p\leq M-1,\\
&\psi_{Ni+q+1/2}^{+(2)}=\psi^+_{(M+N)i+M+q+1/2},\quad \psi_{Ni-q-1/2}^{-(2)}=\psi^-_{(M+N)i-M-q-1/2}, \quad 0\leq q\leq N-1,
\end{align*}
and $S_M^{(a)}=\sum_{j\in\mathbb{Z}+1/2}\psi_j^{+(a)}\otimes\psi_{-j+M}^{-(a)}$, $S_N^{(a)}=\sum_{j\in\mathbb{Z}+1/2}\psi_j^{+(a)}\otimes\psi_{-j+N}^{-(a)}$, we will have $$S_{M+N}=S_M^{(1)}+S_N^{(2)}.$$ Thus \eqref{SMNtaul} will become $(S_M^{(1)}+S_N^{(2)})(\tau_l\otimes\tau_{l'})=0.$ If introduce $\sigma(\tau_l)=\sum_{m\in\mathbb{Z}}(-1)^{\frac{m(m-1)}{2}}
Q_1^{m}Q_2^{l-m}\tau_{m,m-l}(t^{(1)},t^{(2)}),$ then by two--component boson-fermion correspondence, we will have 
\begin{align*}
&\oint_{C_R}\frac{dz}{2\pi i}z^{M+m_1-m_1'}e^{\xi(t^{(1)}-t^{(1)'},z)}\tau_{m_1,m_2}(t-[z^{-1}]_1)\tau_{m_1',m_2'}(t+[z^{-1}]_1)\nonumber\\
=&\oint_{C_r}\frac{dz}{2\pi i}z^{-N+m_2-m_2'}e^{\xi(t^{(2)}-t^{(2)'},z^{-1})}\tau_{m_1+1,m_2+1}(t-[z]_2)\tau_{m_1'-1,m_2'-1}(t+[z]_2),\quad m_1-m_2\geq m_1'-m_2'.
\end{align*}
Now we have obtain the bilinear equations of the $\widehat{sl}_{M+N}$--reduction for the 2dKP hierarchy. In terms of wave function \eqref{psi_1}, we can write \eqref{SMNtaul} into 
\begin{align}\label{ZMZNoint_{C_R}}
\oint_{C_R}\frac{dz}{2\pi iz}z^M\Psi_{1}(\bm{m},t,z)\widetilde{\Psi}_{1}(\bm{m}',t',z)=\oint_{C_r}\frac{dz}{2\pi iz}z^{-N}\Psi_{2}(\bm{m},t,z)\widetilde{\Psi}_{2}(\bm{m}',t',z), \quad m_1-m_2\geq m_1'-m_2'.
\end{align}
\begin{proposition} 
For the 2dKP wave functions $\Psi_a$ satisfying \eqref{ZMZNoint_{C_R}}, if denote $\mathcal{L}=B_M^{(1)}+B_N^{(2)}$, then  
\begin{align*}
  \mathcal{L}(\Psi_1)=z^M\Psi_1, \quad \mathcal{L}(\Psi_2)=z^{-N}\Psi_2,
\end{align*}
where $B_M^{(1)}=(L_1^M)_{1,\geq0},$ $B_N^{(2)}=(L_2^N)_{\Delta_2^*,\geq1}$.
\end{proposition}
\begin{proof}
Firstly by Lemma \ref{ABlemma} and Corollary \ref{coro exist oftaufunction}, we have from \eqref{SMNtaul} that
\begin{align}\label{S1LaMS2La2N}
&\sum_{j_2\in\mathbb{Z}}\left(S_1(\bm{m},\Lambda_1)\Lambda_1^{M}{S}_1^{-1}(\bm{m}+j_2\bm{e}_2,\Lambda_1)\Lambda_1\right)_{1,\leq j_2}\Lambda_2^{j_2}\nonumber\\
=&\sum_{j_1\in\mathbb{Z}}\left(S_2(\bm{m},\Lambda_2)\La_2^{-N}{S}_2^{-1}(\bm{m}+(j_1-1)\bm{e}_1,\Lambda_2)\Delta_2^{*-1}\right)_{2,\geq j_1}\Lambda_1^{j_1}.
\end{align}

If consider the coefficients of $\La_2^0$ of both sides of \eqref{S1LaMS2La2N}, then 
\begin{align*}
  (L_1^M)_{1,<0}&=\Big(S_2(\bm{m},\La_2)\La_2^{-N}\Delta_1^{-1}S_2^{-1}(\bm{m},\La_2)\Delta_2^{*-1}\Big)_{2,[0]}\\
  &=\Big(B_N^{(2)}(\bm{m},\La_2)S_2(\bm{m},\La_2)\Delta_1^{-1}S_2^{-1}(\bm{m},\La_2)\Delta_2^{*-1}\Big)_{2,[0]}.
\end{align*}
Recall from Corollary \ref{corollary pi1pi2}, we know $\pi_1(\La_2^{-k})=(\La_2^{-k}S_2\Delta_1^{-1}S_2^{-1}\Delta_2^{*-1})_{2,[0]}+1$ and $B_N^{(2)}(\bm{m},\La_2)=(L_2^N)_{2,<0}-(L_2^N)_{2,<0}(1),$ thus $(L_1^M)_{1,<0}=\pi_1(B_N^{(2)})$. So we can know
\begin{align*}
\mathcal{L}(\Psi_1)&=(B_M^{(1)}+B_N^{(2)})(\Psi_1)\nonumber\\
  &=(B_M^{(1)}+\pi_1(B_N^{(2)}))(\Psi_1)=L_1^M(\Psi_1)=z^M\Psi_1,
\end{align*}
where we have used $H(\Psi_1)=0$.

Similarly, if we consider the coefficient of $\La_1$ of \eqref{S1LaMS2La2N}, we can know
\begin{align*}
\Big(S_2(\bm{m},\La_2)\La_2^{-N}S_2^{-1}(\bm{m},\La_2)\Delta_2^{*-1}\Big)_{2,\geq1}\La_1&=\sum_{j_2\geq1}\Big(S_1(\bm{m},\La_1)\La_1^MS_1^{-1}(\bm{m}+j_2e_2,\La_1)\La_1\Big)_{1,[1]}\La_2^{j_2}\\
  &=\Big(S_1(\bm{m},\La_1)\La_1^M\Delta_2^{*-1}S_1^{-1}(\bm{m},\La_1)\Big)_{1,[0]}\La_1.
\end{align*}
Notice that $(B_N^{(2)}\Delta_2^*)_{2,\geq1}=0$ and $(L_2^N)_{\Delta_2^*,\leq0}=(L_2^N)_{2,\geq0}+(L_2^N)_{2,<0}(1),$
thus $(L_2^N)_{\Delta_2^*,\leq0}=\pi_2(B_M^{(1)})$ and
\begin{align*}
\mathcal{L}(\Psi_2)&=(B_M^{(1)}+B_N^{(2)})(\Psi_2)\\
&=(\pi_2(B_M^{(1)})+B_N^{(2)})(\Psi_2)=L_2^N(\Psi_2)=z^{-N}\Psi_2.
\end{align*}
\end{proof}
Recall $\Psi_a(\bm{m},t,z)=e^{\xi(t^{(a)},z^{3-2a})}S_a(\bm{m},t,\La_a)(z^{ma}),$ so if denote $$D_{M,N}=\partial_{t_M^{(1)}}+\partial_{t_N^{(2)}},$$ then we have the corollary below.
\begin{corollary}
For the $\widehat{sl}_{M+N}$--reduction of the 2dKP hierarchy,
$$D_{M,N}(S_a)=D_{M,N}(L_a)=D_{M,N}(H)=0.$$
\end{corollary}
\begin{proof}
By $\mathcal{L}(\Psi_1)=z^M\Psi_1$ and Corollary \ref{coro exist oftaufunction}, we know $D_{M,N}(\Psi_1)=z^M\Psi_1$, which means $$(D_{M,N}(S_1)+z^MS_1)(z^{m_1})=z^MS_1(z^{m_1}),$$ that is to say $D_{M,N}(S_1)=0$. Similarly, by $\mathcal{L}(\Psi_2)=z^{N}\Psi_a$, we can know $D_{M,N}(S_2)=0$. So further by $L_1=S_1\La_1S_1^{-1}$, $L_2=S_2\La_2^{-1}S_2^{-1}$ and $H=-\Lambda_1\Lambda_2S_1\Delta_2^{*}S_1^{-1}=-\Lambda_1\Delta_2S_2\Delta_1^{*}S_2^{-1},$ we can get 
\begin{align*}
 D_{M,N}(L_a)=D_{M,N}(H)=0.
\end{align*}

\end{proof}
\begin{corollary}
For the $\widehat{sl}_{M+N}$--reduction of the 2dKP hierarchy, $\mathcal{L}=B_M^{(1)}+B_N^{(1)}$ satisfies 
\begin{align}
  &\partial_{t_k^{(a)}}\mathcal{L}-[B_{a,k},\mathcal{L}]\in\E H,\label{particalL=[B,L]}\\
  &H\mathcal{L}=\widehat{\mathcal{L}}H,\label{HL=WIDELH}
\end{align}
where $\widehat{\mathcal{L}}=C_M^{(1)}+C_N^{(2)}=B_M^{(2)}(\bm{m}+\bm{e})+\Delta_2^*B_N^{(2)}(\bm{m}+\bm{e})\Delta_2^{*-1}.$
\end{corollary}
\begin{proof}
\eqref{particalL=[B,L]} can be proved by Proposition \ref{bk_1,bl_2} and $D_{M,N}(B_k^{(a)})=0$ derived by $D_{M,N}(L_a)=0$,
while \eqref{HL=WIDELH} can be obtained by $D_{M,N}(H)=\widehat{\mathcal{L}}H-H\mathcal{L}=0$.
\end{proof}
By the fact that if $A(\Psi_a)=0$ for $A\in\mathcal{B}(\La_a^{2a-3})$ with $a=1 \ \text{or}\ 2$, then $A=0$, we have the following corollary.
\begin{corollary}
For the $\widehat{sl}_{M+N}$--reduction of the 2dKP hierarchy,
$$\pi_1(\mathcal{L})=L_1^{M}, \quad\pi_2(\mathcal{L})=L_2^N.$$
\end{corollary}

Therefore if assume that $\mathcal{L}$ has the following form,
$$\mathcal{L}=\La_1^M+\sum_{k=0}^{M-1}u_k\La_1^k+\sum_{l=1}^{N}v_l(\La_2^{-l}-1),$$
then 
\begin{align*}
\pi_1(\mathcal{L})=&\La_1^M+\sum_{k=0}^{M-1}u_k\La_1^k+\sum_{l=1}^{N}v_l\left(\prod_{j=1}^l\Big(1-\iota_{\Lambda_1^{-1}}
\big(\La_1-\rho(\bm{m}-j\bm{e}_2)\big)^{-1}\cdot\rho(\bm{m}-j\bm{e}_2)\Big)-1\right),\\
\pi_2(\mathcal{L})=&(-1)^M\prod_{j=1}^M\Big(\iota_{\Lambda_2}(\La_2-1)^{-1}\cdot\rho(\bm{m}+(j-1)\bm{e}_1)\Big)\\
&+\sum_{k=0}^{M-1}u_k(-1)^k\prod_{j=1}^k\Big(\iota_{\Lambda_2}(\La_2-1)^{-1}\cdot\rho(\bm{m}+(j-1)\bm{e}_1)\Big)+\sum_{l=1}^{N}v_l(\La_2^{-l}-1).
\end{align*}
Notice that there are $M+N+1$ independent unknown functions in the $\widehat{sl}_{M+N}$--reduction of the 2dKP hierarchy, including $u_k(0\leq k\leq M-1)$, $v_l(1\leq l\leq N)$ and $\rho$, where the corresponding evolution equations are given by
\begin{align*}
  &H_{t_k^{(a)}}=C_k^{(a)}H-HB_k^{(a)},\\ &\partial_{t_k^{(a)}}\mathcal{L}=[\pi_1(B_k^{(a)}),\pi_1(\mathcal{L})]_{1,\geq0}+[\pi_2(B_k^{(a)}),\pi_2(\mathcal{L})]_{\Delta_2^*,\geq1},\\ \text{or} \quad&\partial_{t_k^{(a)}}\pi_b(\mathcal{L})=[\pi_b(B_k^{(a)}),\pi_b(\mathcal{L})],\quad a,b=1,2.
\end{align*}
\noindent\textbf{Example. }
One can find that when $M=N=1$,
\begin{align*}
&\pi_1(\mathcal{L}(\bm{m}))=\La_1+u_0(\bm{m})-v_1(\bm{m})\La_1\cdot\big(\La_1-\rho(\bm{m}-\bm{e})\big)^{-1},\\
&\pi_2(\mathcal{L}(\bm{m}))=v_1(\bm{m})\Delta_2^*+u_0(\bm{m})-\rho(\bm{m})-\Delta_2^{*-1}\rho(\bm{m}).
\end{align*}

 Let us give a summary of the $\widehat{sl}_{M+N}$--reduction of the 2dKP hierarchy. Given $H=\La_1\Delta_2+\rho$ and $$\mathcal{L}=\La_1^M+\sum_{k=0}^{M-1}u_k\La_1^k+\sum_{l=1}^{N}v_l(\La_2^{-l}-1),$$ 
 if define 
 \begin{align*}
 &B_k^{(1)}(\bm{m})=\Big(\pi_1(\mathcal{L}(\bm{m}))^{\frac{k}{M}}\Big)_{1,\geq 0}, B_k^{(2)}(\bm{m})=\Big(\pi_2(\mathcal{L}(\bm{m}))^{\frac{k}{N}}\Big)_{\Delta_2^*,\geq1},\\
 &C_k^{(1)}(\bm{m})=B_k^{(1)}(\bm{m}+\bm{e}),\quad C_k^{(2)}(\bm{m})=\Delta_2^*B_k^{(2)}(\bm{m}+\bm{e})\Delta_2^{*-1},
 \end{align*}
then
\begin{align*}
  &H\mathcal{L}=(C_M^{(1)}+C_N^{(2)})H,\quad \partial_{t_k^{(a)}}H=C_k^{(a)}H-HB_k^{(a)},\\
  &\partial_{t_k^{(a)}}\mathcal{L}=[\pi_1(B_k^{(a)}),\pi_1(\mathcal{L})]_{1,\geq0}+[\pi_2(B_k^{(a)}),\pi_2(\mathcal{L})]_{\Delta_2^*,\geq1}.
\end{align*}

\section{Conclusions and Discussions}
In this paper, we firstly investigate several equivalent formulation of the 2dKP hierarchy including
\begin{itemize}
\item The bilinear equation in terms of tau function:
\begin{align*}
&\oint_{C_R}\frac{dz}{2\pi i}z^{m_1-m_1'}e^{\xi(t^{(1)}-t^{(1)'},z)}\tau_{m_1,m_2}(t-[z^{-1}]_1)\tau_{m_1',m_2'}(t+[z^{-1}]_1)\nonumber\\
=&\oint_{C_r}\frac{dz}{2\pi i}z^{m_2-m_2'}e^{\xi(t^{(2)}-t^{(2)'},z^{-1})}\tau_{m_1+1,m_2+1}(t-[z]_2)\tau_{m_1'-1,m_2'-1}(t+[z]_2),\quad m_1-m_2\geq m_1'-m_2'.
\end{align*}
\item The bilinear equation by wave functions and adjoint wave functions:
\begin{align*}
\oint_{C_R}\frac{dz}{2\pi iz}\Psi_{1}(\bm{m},t,z)\widetilde{\Psi}_{1}(\bm{m}',t',z)=\oint_{C_r}\frac{dz}{2\pi iz}\Psi_{2}(\bm{m},t,z)\widetilde{\Psi}_{2}(\bm{m}',t',z), \quad m_1-m_2\geq m_1'-m_2'.
\end{align*}
\item Wave operators $S_a$ and $\widetilde{S}_a$:
\begin{align*}
&S_1(\bm{m},t,\Lambda_1)=1+\sum_{k=1}^{+\infty}a_k^{(1)}\Lambda_{1}^{-k},\quad\quad\quad\ \widetilde{S}_1(\bm{m},t,\Lambda_1)=1+\sum_{k=1}^{+\infty}\widetilde{a}_k^{(1)}\Lambda_{1}^k,\\ &S_2(\bm{m},t,\Lambda_2)=a_0^{(2)}+\sum_{k=1}^{+\infty}a_k^{(2)}\Lambda_{2}^k,\quad\quad \widetilde{S}_2(\bm{m},t,\Lambda_2)=\widetilde{a}_0^{(2)}+\sum_{k=1}^{+\infty}\widetilde{a}_k^{(2)}\Lambda_{2}^{-k}.
\end{align*}
\item Lax equation:
\begin{align*}
  &\partial_{t_k^{(a)}}L_b(\bm{m},\Lambda_b)=[\pi_b(B_k^{(a)}(\bm{m},\Lambda_a)),L_b(\bm{m},\Lambda_b)],\\  &H(\bm{m})L_1(\bm{m},\Lambda_1)=L_1(\bm{m}+\bm{e},\Lambda_1)H(\bm{m}),\quad H(\bm{m})L_2(\bm{m},\Lambda_2)=\Delta_2^{*}L_2(\bm{m}+\bm{e})\Delta_2^{*-1}H(\bm{m}),\\ &\p_{t_k^{(a)}}H(\bm{m})=C_k^{(a)}(\bm{m},\Lambda_a)H(\bm{m})-H(\bm{m})\cdot B_k^{(a)}(\bm{m},\Lambda_a),\ a,b=1,2.
\end{align*} 
\end{itemize}
Then the $\widehat{sl}_{M+N}$--reduction of the 2dKP hierarchy are also discussed, which can be described by the operator $H=\La_1\Delta_2+\rho$ and the Lax operator $\mathcal{L}=\La_1^M+\sum_{k=0}^{M-1}u_k\La_1^k+\sum_{l=1}^{N}v_l(\La_2^{-l}-1)$ satisfying $\mathcal{L}(\Psi_1)=z^{M}\Psi_1$ and $\mathcal{L}(\Psi_2)=z^{-N}\Psi_2$.

For the 2dKP hierarchy, the first Lax operator $L_1=\La_1+\sum_{i=0}^{+\infty}u_i^{(1)}\Lambda_1^{-i}$ satisfying $\partial_{t_k^{(1)}}L_1=[(L_1^k)_{1,\geq0},L_1]$, which is just the one component discrete KP hierarchy \cite{Adler1999-1,Dickey1999LMP}, while the second Lax operator $L_2=u_{-1}^{(2)}\La_2^{-1}+\sum_{i=0}^{+\infty}u_i^{(2)}\La_2^i$ satisfying $\partial_{t_k^{(2)}}L_2=[(L_2^k)_{\Delta_2^*,\geq1},L_2]$, which can be viewed as the modified discrete KP hierarchy \cite{Chengjipeng2019}. Thus the 2dKP hierarchy contains the information of the discrete KP hierarchy and the modified discrete KP hierarchy.
 
The results here provide one typical example of the derivation of the Lax equation from the bilinear equation, and may be helpful in understanding the Shiota construction \cite{Cui2025,Shiota1989} of the Lax formulations of the integrable hierarchies. Also, we believe that there should be some applications for the 2dKP hierarchy in soliton theory and integrable system, especially in the KP reduction method \cite{Ohta2011,Yang2022}.\\

\noindent{\bf Acknowledgements}:

This work is supported by National Natural Science Foundation of China (Grant Nos. 12171472 and 12261072).\\

\noindent{\bf Conflict of Interest}:

 The authors have no conflicts to disclose.\\

\noindent{\bf Data availability}:

Data sharing is not applicable to this article as no new data were created or analyzed in this study.

\end{document}